\documentclass{lmcs} 
\pdfoutput=1

\usepackage{lastpage}
\lmcsdoi{18}{2}{19}
\lmcsheading{}{\pageref{LastPage}}{}{}%
{Apr.~12,~2021}{Jun.~15,~2022}{}


\newcommand{\myparagraph}[1]{\bigskip\noindent\textbf{#1}\;\,}

\newcommand{\into}{\hookrightarrow}

\newcommand{\myby}[1]{(\text{#1})}
\newcommand{\mybys}[2]{(\text{#1})\;\text{and}\;(\text{#2})}
\newcommand{\Rat}{\mathbb{Q}\cap[0,1]}

\newcommand{\Sem}[1]{\llbracket{#1}\rrbracket}

\newcommand{\Pow}{\mathcal{P}}
\newcommand{\Pfin}{\Pow_\omega}

\newcommand{\id}{\mathsf{id}}

\newcommand{\Set}{\mathsf{Set}}
\newcommand{\HMet}{\mathsf{HMet}}

\newcommand{\op}{^{\mathrm{op}}}

\newcommand{\unithemi}{d_{\ominus}}

\newcommand{\nonexp}[2]{#1 \rightarrow_1 #2}
\newcommand{\supnorm}[1]{\lVert #1 \rVert_\infty}

\newcommand{\dfin}{\dfun_\omega}
\newcommand{\dfun}{\mathcal{D}}
\newcommand{\ffun}{\mathcal{F}}
\newcommand{\qfun}{\mathcal{Q}}
\newcommand{\mfun}{\mathcal{M}}
\newcommand{\nfun}{\mathcal{N}}

\newcommand{\frel}[2]{#1 \mathbin{\ooalign{$\rightarrow$\cr$\hspace{0.12ex}+$\cr}} #2}
\newcommand{\rev}[1]{#1 ^ \circ}
\newcommand{\gr}[1]{\text{Gr}_{#1}}

\newcommand{\modf}[1]{\mathcal{L}_{#1}}
\newcommand{\cpl}[2]{\mathsf{Cpl}(#1,#2)}
\newcommand{\cpltri}[3]{\mathsf{Cpl}(#1,#2,#3)}
\newcommand{\cplind}[3]{\mathsf{Cpl}_{#1}(#2,#3)}
\newcommand{\constfun}[2]{{#1}_{#2}}
\newcommand{\zerofun}[1]{\constfun{0}{#1}}
\newcommand{\Conv}{\mathcal{C}}

\newcommand{\swap}{\mathsf{swap}}

\newcommand{\expect}[1]{\mathbb{E}_{#1}}

\DeclareMathOperator{\convx}{conv}
\newcommand{\last}{\mathsf{last}}
\newcommand{\app}{\mathsf{app}}

\newcommand{\strength}[2]{\mathsf{str}_{{#1},{#2}}}

\newcommand{\itemref}[1]{(\ref{#1})}
\newcommand{\elem}{{\in}}

\usepackage[author=anonymous,nomargin,marginclue,footnote,draft]{fixme}
\FXRegisterAuthor{ls}{als}{LS}
\FXRegisterAuthor{pw}{apw}{PW}
\usepackage{amsmath}
\usepackage{amsthm}
\usepackage{MnSymbol}
\usepackage{stmaryrd}
\usepackage{bussproofs}
\usepackage[all]{xy}
\usepackage{mathtools}
\usepackage{xparse}
\usepackage{tikz}
\usepackage{hyperref}
\usepackage[utf8]{inputenc}

\theoremstyle{definition}\newtheorem{defn}[thm]{Definition}
\theoremstyle{definition}\newtheorem{expl}[thm]{Example}

\begin{document}

\pdfstringdefDisableCommands{%
  \def\\{}%
}
\title[Characteristic Logics for Behavioural Hemimetrics]
{Characteristic Logics for Behavioural Hemimetrics\texorpdfstring{\\}{ }via Fuzzy Lax Extensions}

\titlecomment{Work forms part of the DFG-funded project
  \emph{Probabilistic description logics as a fragment of
    probabilistic first-order logic} (SCHR 1118/6-2)}

\author[P.~Wild]{Paul Wild\lmcsorcid{0000-0001-9796-9675}}
\author[L.~Schr\"oder]{Lutz Schr\"oder\lmcsorcid{0000-0002-3146-5906}}
\address{Friedrich-Alexander-Universit\"at Erlangen-N\"urnberg}
\email{paul.wild@fau.de, lutz.schroeder@fau.de}

\begin{abstract}
  In systems involving quantitative data, such as probabilistic,
  fuzzy, or metric systems, behavioural distances provide a more
  fine-grained comparison of states than two-valued notions of
  behavioural equivalence or behaviour inclusion. Like in the
  two-valued case, the wide variation found in system types creates a
  need for generic methods that apply to many system types at
  once. Approaches of this kind are emerging within the paradigm of
  universal coalgebra, based either on lifting pseudometrics along set
  functors or on lifting general real-valued (\emph{fuzzy}) relations
  along functors by means of \emph{fuzzy lax extensions}. An immediate
  benefit of the latter is that they allow bounding behavioural
  distance by means of fuzzy (bi-)simulations that need not themselves
  be hemi- or pseudometrics; this is analogous to classical
  simulations and bisimulations, which need not be preorders or
  equivalence relations, respectively. The known generic pseudometric
  liftings, specifically the generic Kantorovich and Wasserstein
  liftings, both can be extended to yield fuzzy lax extensions, using
  the fact that both are effectively given by a choice of quantitative
  modalities. Our central result then shows that in fact all fuzzy lax
  extensions are Kantorovich extensions for a suitable set of
  quantitative modalities, the so-called \emph{Moss modalities}. For
  \emph{nonexpansive} fuzzy lax extensions, this allows for the
  extraction of quantitative modal logics that characterize
  behavioural distance, i.e.\ satisfy a quantitative version of the
  Hennessy-Milner theorem; equivalently, we obtain expressiveness of a
  quantitative version of Moss' coalgebraic logic. All our results
  explicitly hold also for asymmetric distances (\emph{hemimetrics}),
  i.e.~notions of quantitative simulation.
\end{abstract}

\maketitle

\section{Introduction}
Branching-time equivalences on reactive systems are typically governed
by notions of \emph{bisimilarity}~\cite{Park81,Milner89}. For systems
involving quantitative data, such as transition probabilities, fuzzy
truth values, or labellings in metric spaces, it is often appropriate
to use more fine-grained, \emph{quantitative} measures of behavioural
similarity, arriving at notions of \emph{behavioural
  distance}. Distance-based approaches in particular avoid the problem
that small quantitative deviations in behaviour will typically render
two given systems inequivalent under two-valued notions of
equivalence, losing information about their similarity. We note in
passing that behavioural distances are typically \emph{pseudo}metrics,
i.e.\ distinct states can have distance~$0$ if their behaviours are
exactly equivalent.

Behavioural distances serve evident purposes in system verification,
allowing as they do for a reasonable notion of a specification being
satisfied up to an acceptable margin of deviation
(e.g.~\cite{Gavazzo18}). Applications have also been proposed in
differential privacy \cite{ChatzikokolakisEA14} and conformance
testing of hybrid systems~\cite{Khakpour2015NotionsOC}. Like their
two-valued counterparts, behavioural distances have been introduced
for quite a range of system types, such as various forms of
probabilistic labelled transition systems or labelled Markov processes
\cite{GiacaloneEA90,bw:behavioural-pseudometric,d:labelled-markov-processes,
  dgjp:metrics-labelled-markov}; systems combining nondeterministic
and probabilistic branching variously known as nondeterministic
probabilistic transition systems~\cite{cgt:logical-bisim-metrics},
probabilistic automata~\cite{DengEA06}, and Markov decision
processes~\cite{FernsEA04}; weighted automata~\cite{BalleEA17}; fuzzy
transition systems~\cite{CaoEA13} and fuzzy Kripke
models~\cite{Fan15}; and various forms of metric transition systems
\cite{afs:linear-branching-metrics,FahrenbergEA11,FahrenbergLegay14},
which are non-deterministic transition systems with additional
quantitative information, e.g.\ a metric on the labels and/or the
states. Besides symmetric notions of behavioural distance, there are
asymmetric variants, which correspond to quantitative notions of
simulation, e.g.~for rational-~\cite{CernyEA12},
real-~\cite{ThraneEA10}, and lattice-weighted transition
systems~\cite{PanEA15}.

This range of variation creates a need for unifying concepts and
methods. The present work contributes to developing such a unified
view within the framework of universal coalgebra, which is based on
abstracting a wide range of system types (including all the mentioned
ones) as set functors.  Specifically, we work with a generic notion of
\emph{quantitative simulation} via the key notion of
\emph{nonexpansive (fuzzy) lax extension} of a functor. Fuzzy and
quantale-valued generalizations of lax extensions have been studied in
the past~\cite{Gavazzo18,HofmannEA14}; we identify a new criterion for
such lax extensions to be \emph{nonexpansive} (equivalently
\emph{strong} in the sense of Gavazzo~\cite{Gavazzo18}) that allows us
to relate lax extensions to fuzzy logics featuring nonexpansive
modalities via a Hennessy-Milner thoerem. Given a fuzzy lax extension,
behavioural distance is defined as the greatest quantitative
simulation; in general, behavioural distance is a \emph{hemimetric},
i.e.\ obeys the usual axioms of a pseudometric except symmetry, or
equivalently a generalized metric space in the sense of
Lawvere~\cite{Lawvere73}.

For instance, on weighted transition systems with labels in a finite
metric space $(M,d_M)$ Larsen et al.~\cite{LarsenEA11} consider a
\emph{simulation distance} defined as the least fixed point of the
equation
\begin{equation*}
  d(s,t) = \adjustlimits\sup_{s\xrightarrow{m}s'}\inf_{t\xrightarrow{n}t'}
    d_M(m,n) + \lambda d(s',t'),
\end{equation*}
where $0\le\lambda < 1$ is a discount factor. We shall later see that
this simulation distance arises via a nonexpansive fuzzy lax extension
and thus forms an instance of this framework.

For lax extensions
obeying a suitable symmetry axiom, quantitative simulations are in
fact quantitative bisimulations in the sense that their relational
converse is also a simulation, and the induced behavioural distance is
symmetric, i.e.~forms a pseudometric. Existing coalgebraic approaches
to behavioural pseudometrics rely on pseudometric liftings of
functors~\cite{bbkk:coalgebraic-behavioral-metrics}, and in particular
lift only pseudometrics; contrastingly, fuzzy lax extensions act on
unrestricted quantitative relations.  Hence, quantitative
(bi-)simulations need not themselves be hemi- or pseudometrics, in
analogy to classical bisimulations not needing to be equivalence
relations, and thus may serve as small certificates for low
behavioural distance. We show that two known systematic constructions
of functor liftings from chosen sets of modalities, the generic
Wasserstein and Kantorovich liftings, both extend to yield
nonexpansive fuzzy lax extensions (it is essentially known that the
Wasserstein lifting yields a fuzzy lax extension~\cite{Hofmann07}). As
our main result, we then establish that \emph{every} fuzzy lax
extension of a finitary functor is a Kantorovich extension induced by
a suitable set of modalities, the so-called Moss modalities. Notably,
the definition of the Moss modalities involves application of the
given lax extension to the quantitative elementhood relation, and
hence centrally relies on lifting quantitative relations that fail to
be hemi- or pseudometrics. 

This result may be seen as a quantitative version of previous results
asserting the existence of separating sets of two-valued modalities
for finitary functors~\cite{Schroder08,KurzLeal09,MartiVenema15},
which allow for generic Hennessy-Milner-type theorems stating that
states in finitely branching systems (coalgebras) are behaviourally
equivalent iff they satisfy the same modal
formulae~\cite{Pattinson04,Schroder08}. Indeed, for nonexpansive lax
extensions our main result similarly allows \emph{extracting
characteristic quantitative modal logics} from given behavioural hemi-
or pseudometrics, where a logic is \emph{characteristic} or
\emph{expressive} if the induced logical distance of states coincides
with behavioural distance. This result may equivalently be phrased as
expressiveness of a quantitative version of Moss' coalgebraic
logic~\cite{Moss99}, which provides a coalgebraic generalization of
the classical relational $\nabla$-modality (which e.g.\ underlies the
$a\to\Psi$ notation used in Walukiewicz's $\mu$-calculus completeness
proof~\cite{Walukiewicz95}). We relax the standard requirement of
finite branching, i.e.\ use of finitary functors, to an
approximability condition called \emph{finitary separability}, and
hence in particular cover countable probabilistic branching. Moreover,
we emphasize that we obtain characteristic logics also for asymmetric
distances, i.e.~notions of quantitative simulation.

\paragraph*{Organization} We recall basic concepts on hemi- and
pseudometrics, coalgebraic bisimilarity, and coalgebraic logic in
Section~\ref{sec:prelim}. The central notion of (nonexpansive) fuzzy
lax extension is introduced in Section~\ref{sec:fuzzy-rel}, and the
arising principle of quantitative (bi-)simulation in
Section~\ref{sec:bisim}. The generic Kantorovich and Wasserstein
liftings are discussed in Sections~\ref{sec:kantorovich} and
\ref{sec:wasserstein}, respectively. Our central result showing that
every lax extension is a Kantorovich lifting is established in
Section~\ref{sec:lax-kantorovich}. In Section~\ref{sec:cml}, we show
how our results amount to extracting characteristic modal logics from
given nonexpansive lax extensions.

\paragraph*{Related Work} Probabilistic quantitative characteristic
modal logics go back to Desharnais et
al.~\cite{dgjp:metrics-labelled-markov}; they relate to fragments of
quantitative
$\mu$-calculi~\cite{HuthKwiatkowska97,ms:lukasiewicz-mu,mm:prob-calculus-expectations}.
A further well-known class of quantitative modal logics are fuzzy
modal and description logics
(e.g.~\cite{Morgan79,Fitting91,Straccia98,LukasiewiczStraccia08}). Van
Breugel and Worrell~\cite{bw:behavioural-pseudometric} prove a
Hennessy-Milner theorem for quantitative probabilistic modal logic.
Quantitative Hennessy-Milner-type
theorems 
have since been established for fuzzy modal logic with G\"odel
semantics~\cite{Fan15}, for systems combining probability and
non-determinism~\cite{DuEA16}, and for Heyting-valued modal
logics~\cite{EleftheriouEA12} as introduced by
Fitting~\cite{Fitting91}. König and
Mika-Michalski~\cite{km:metric-bisimulation-games} provide a
quantitative Hennessy-Milner theorem in coalgebraic generality for the
case where behavioural distance is induced by the pseudometric
Kantorovich lifting defined by the same set of modalities as the
logic, a result that we complement by showing that in fact all fuzzy
lax extensions are Kantorovich. The assumptions of König and
Mika-Michalski's theorem require that behavioural distance be
approximable in~$\omega$ steps. We give a sufficient criterion for
this property: The predicate liftings need to be nonexpansive, and
the given lax extension needs to be finitarily separable (as mentioned
above). Again, we remove any assumption of symmetry, obtaining an
expressiveness criterion for characteristic logics of quantitative
simulations; in this sense, our work relates also to (coalgebraic and
specific) results on characteristic logics for two-valued notions of
similarity, e.g.~\cite{Glabbeek90,Baltag00,Cirstea06,KapulkinEA12,FordEA21}.

Fuzzy lax extensions are a quantitative version of lax
extensions~\cite{MartiVenema15,ThijsThesis,Levy11,BackhouseEA91}, which in turn
belong to an extended strand of research on relation
liftings~\cite{HughesJacobs04,ThijsThesis,Levy11}.
They appear to go back to work on monoidal
topology~\cite{HofmannEA14}, and have been used in work on applicative
bisimulation~\cite{Gavazzo18}; as indicated above,
Hofmann~\cite{Hofmann07} effectively already introduces the generic
Wasserstein lax extension (without using the term but proving the
relevant properties, except nonexpansiveness). Our notion of
\emph{nonexpansive} lax extension, which is central to the connection
with characteristic logics, appears to be new, but as indicated above
it can be seen to relate to a condition involving the strength of the
underlying functor as considered by Gavazzo~\cite{Gavazzo18}. Our
method of extracting quantitative modalities from fuzzy lax extensions
generalizes the construction of two-valued Moss liftings for
(two-valued) lax extensions~\cite{KurzLeal09,MartiVenema15}.

This paper is an extended and revised version of a previous
conference publication~\cite{ws:fuzzy-lax}. Besides containing
additional discussion and full proofs, the present version generalizes
the overall technical treatment including the main results to the
asymmetric setting, thus covering not only quantitative notions of
bisimulation but also quantitative notions of simulation.

\section{Preliminaries}\label{sec:prelim}

\noindent We recall basic notions on metrics, pseudometrics (where
distinct points may have distance~$0$), and hemimetrics (where
additionally distance is not required to be symmetric). Moreover, we
give a brief introduction to universal coalgebra~\cite{Rutten00} and
the generic treatment of two-valued bisimilarity. Basic knowledge of
category theory (e.g.~\cite{AdamekHerrlich90}) will be helpful.

\paragraph*{Hemimetric Spaces} For the present purposes, we are
interested only in bounded distance functions, and then normalize
distances to lie in the unit interval. Thus, a \emph{($1$-bounded)
  hemimetric} on a set~$X$ is a function
\begin{math}
  d\colon X\times X\to[0,1]
\end{math}
satisfying $d(x,x) =0$ (reflexivity), and $d(x,z)\le d(x,y)+d(y,z)$
(triangle inequality) for $x,y,z\in X$. If additionally
$d(x,y) =d(y,x)$ for all $x,y\in X$ (symmetry), then~$d$ is a
\emph{pseudometric}. If moreover for all $x,y\in X$, $d(x,y)=0$
implies $x=y$, then~$d$ is a \emph{metric}. The pair $(X,d)$ is a
\emph{hemimetric space}, or respectively a
\emph{\mbox{(pseudo-)}metric space} if $d$ is a
\mbox{(hemi-/pseudo-)}metric. We write $\ominus$ for truncated subtraction
on the unit interval, i.e.\ $x\ominus y=\max(x-y,0)$ for
$x,y\in[0,1]$.  Then $\unithemi(x,y)=x\ominus y$ defines a hemimetric
$\unithemi$ on $[0,1]$; moreover, $[0,1]$ is a metric space under
Euclidean distance $d_E(x,y)=|x-y|$. The \emph{supremum distance} of
functions $f,g\colon X\to[0,1]$ is
$\supnorm{f-g} = \sup_{x\in X} |f(x)-g(x)|$. A map $f\colon X\to Y$ of
hemimetric spaces $(X,d_1)$, $(Y,d_2)$, is \emph{nonexpansive}
(notation: $f\colon\nonexp{(X,d_1)}{(Y,d_2)}$) if
\begin{math}
  d_2(f(x),f(y))\le d_1(x,y)
\end{math}
for all $x,y\in X$.

\myparagraph{Universal Coalgebra} is a uniform framework for a broad
range of state-based system types. It is based on encapsulating the
transition type of a system as an \mbox{(endo-)}functor, for the
present purposes on the category of sets: A \emph{functor}~$T$ assigns
to each set~$X$ a set~$TX$, and to each map $f\colon X\to Y$ a map
$Tf\colon TX\to TY$, preserving identities and composition. We may
think of~$TX$ as a parametrized datatype; e.g.\ the \emph{(covariant)
  powerset functor}~$\Pow$ assigns to each set~$X$ its powerset
$\Pow X$, and to $f\colon X\to Y$ the direct image map
$\Pow f\colon \Pow X\to \Pow Y, A\mapsto f[A]$; and the
\emph{distribution functor} $\dfun$ maps each set~$X$ to the set of
discrete probability distributions on~$X$. Recall that a discrete
probability distribution on~$X$ is given by a \emph{probability mass
  function} $\mu\colon X\to[0,1]$ such that $\sum_{x\in X}\mu(x)=1$
(implying that the \emph{support} $\{x\in X\mid \mu(x)>0\}$ of~$\mu$
is at most countable); we abuse~$\mu$ to denote also the induced
probability measure, writing $\mu(A)=\sum_{x\in A}\mu(x)$ for
$A\subseteq X$. Moreover,~$\dfun$ maps $f\colon X\to Y$ to
$\dfun f\colon \dfun X\to\dfun Y$, $\mu\mapsto \mu f^{-1}$ where the
\emph{image measure} $\mu f^{-1}$ is given by
$\mu f^{-1}(B)=\mu(f^{-1}[B])$ for $B\subseteq Y$. We will introduce
further examples later.

Systems of a transition type~$T$ are then cast as
\emph{$T$-coalgebras} $(A,\alpha)$, consisting of a set~$A$ of
\emph{states} and a \emph{transition function} $\alpha\colon A\to TA$,
thought of as assigning to each state a structured collection of
successors. E.g.\ a $\Pow$-coalgebra $\alpha\colon A\to\Pow A$ assigns
to each state~$a$ a set $\alpha(a)$ of successors, so is just a
(non-deterministic) transition system. Similarly, a $\dfun$-coalgebra
assigns to each state a distribution over successor states, and thus
is a probabilistic transition system or a Markov chain. A
\emph{morphism} $f\colon (A,\alpha)\to(B,\beta)$ of $T$-coalgebras
$(A,\alpha)$ and $(B,\beta)$ is a map $f\colon A\to B$ such that
$\beta\circ f=Tf\circ\alpha$, where $\circ$ denotes the usual
(applicative) composition of functions; e.g.\ morphisms of
$\Pow$-coalgebras are functional bisimulations, also known as
\emph{$p$-morphisms} or \emph{bounded morphisms}. 

A functor~$T$ is \emph{finitary} if for each set~$X$ and each
$t\in TX$, there exists a finite subset $Y\subseteq X$ such that
$t=Ti(t')$ for some $t'\in TY$, where $i\colon Y\to X$ is the
inclusion map (this is equivalent to the more categorically phrased
condition that~$T$ preserves directed colimits). Intuitively,~$T$ is
finitary if every element of $TX$ mentions only finitely many elements
of~$X$. Every set functor~$T$ has a \emph{finitary part} $T_\omega$
given by
\begin{equation*}
T_\omega X=\bigcup\{Ti[TY]\mid Y\subseteq X\text{ finite},i\colon
Y\to X\text{ inclusion}\}.
\end{equation*}
E.g.~$\Pfin$, the \emph{finite powerset functor}, maps a set to the
set of its finite subsets, and~$\dfin$, the \emph{finite distribution
  functor}, maps a set~$X$ to the set of discrete probability
distributions on~$X$ with finite support. Coalgebras for finitary
functors generalize finitely branching systems, and hence feature in
Hennessy-Milner type theorems, which typically fail under infinite
branching.

\paragraph*{Bisimilarity and Lax Extensions} Coalgebras come with
a canonical notion of observable equivalence: States $a\in A$,
$b\in B$ in $T$-coalgebras $(A,\alpha)$, $(B,\beta)$ are
\emph{behaviourally equivalent} if 
there exist a coalgebra $(C,\gamma)$ and morphisms
$f\colon (A,\alpha)\to (C,\gamma)$, $g\colon (B,\beta)\to(C,\gamma)$
such that $f(a)=g(b)$. Behavioural equivalence can often be
characterized in terms of bisimulation relations, which may provide
small witnesses for behavioural equivalence of states and in
particular need not form equivalence
relations.
The most general known way of treating bisimulation coalgebraically is
via \emph{lax extensions}~$L$ of the functor~$T$, which map relations
$R\subseteq X\times Y$ to
\begin{math}
  LR\subseteq TX\times TY
\end{math}
subject to a number of axioms (monotonicity, preservation of
relational converse, lax preservation of composition, extension of
function graphs)~\cite{MartiVenema15};~$L$ \emph{preserves diagonals}
if $L\Delta_X=\Delta_{TX}$ for each set~$X$, where for any set~$Y$,
$\Delta_Y$ denotes the \emph{diagonal} $\{(y,y)\mid y\in Y\}$. The
\emph{Barr extension}~$\overline T$ of~$T$~\cite{Barr70,Trnkova80} is
defined by
\begin{equation*}
  \overline T R = \{ (T\pi_1 (r), T\pi_2 (r)) \mid r\in TR \}
\end{equation*}
for $R\subseteq X\times Y$, where $\pi_1\colon R \to X$ and
$\pi_2\colon R \to Y$ are the projections;~$\overline T$ preserves
diagonals, and is a lax extension if~$T$ preserves weak
pullbacks. E.g., the Barr extension~$\overline\Pow$ of the powerset
functor~$\Pow$ is the well-known Egli-Milner extension, given by
\begin{equation*}
  (V,W)\in\overline\Pow(R)\iff (\forall x\in V.\,\exists y\in W.\,(x,y)\in R)\land
  (\forall y\in W.\,\exists x\in V.\,(x,y)\in R)
\end{equation*}
for $R\subseteq X\times Y$, $V\in\Pow(X)$, $W\in\Pow(Y)$.  An
\emph{$L$-bisimulation} between $T$-coalgebras $(A,\alpha)$,
$(B,\beta)$ is a relation $R\subseteq A\times B$ such that
\begin{math}
  (\alpha(a),\beta(b))\in LR
\end{math}
for all $(a,b)\in R$; e.g.\ for $L=\overline\Pow$, $L$-bisimulations
are precisely Park/Milner bisimulations on transition systems.  If a
lax extension~$L$ preserves diagonals, then two states are
behaviourally equivalent iff they are related by some
$L$-bisimulation~\cite{MartiVenema15}.

\myparagraph{Coalgebraic Logic} serves as a generic framework for the
specification of state-based systems~\cite{CirsteaEA11b}. For our
present purposes, we are primarily interested in its simplest
incarnation as a modal next-step logic, dubbed \emph{coalgebraic modal
  logic}, and its role as a characteristic logic for behavioural
equivalence in generalization of Hennessy-Milner
logic~\cite{HennessyMilner85}. We briefly recall the syntax and
semantics of coalgebraic modal logic, as well as basic results.  The
framework is based on interpreting custom \emph{modalities} of given
finite arity over coalgebras for a functor~$T$ as $n$-ary
\emph{predicate liftings}, which are families of maps
\begin{equation*}
  \lambda_X\colon(2^X)^n\to 2^{TX}
\end{equation*}
(subject to a naturality condition) where $2=\{\bot,\top\}$ and for
any set~$Y$, $2^Y$ is the set of $2$-valued predicates on~$Y$. We do
not distinguish notationally between modalities and the associated
predicate liftings. Satisfaction of a formula of the
form~$\lambda(\phi_1,\dots,\phi_n)$ (in some ambient logic) in a state
$a\in A$ of a $T$-coalgebra $(A,\alpha)$ is then defined inductively
by
\begin{equation}\label{eq:sem-modal}
  a\models\lambda(\phi_1,\dots,\phi_n)\text{ iff }\alpha(a)\in\lambda_A(\Sem{\phi_1},\dots,\Sem{\phi_n})
\end{equation}
where for any formula~$\psi$,
$\Sem{\psi}=\{c\in A\mid c\models\psi\}$. E.g.\ the standard diamond
modality $\Diamond$ is interpreted over the powerset functor~$\Pow$ by
the predicate lifting
\begin{math}
  \Diamond_X(Y)=\{Z\in\Pow(X)\mid \exists x\in Z.\, Y(x)=\top\},
\end{math}
which according to~\eqref{eq:sem-modal} induces precisely the usual
semantics of~$\Diamond$ over transition systems
($\Pow$-coalgebras). The standard Hennessy-Milner theorem is
generalized coalgebraically~\cite{Pattinson04,Schroder08} as saying
that two states in $T$-coalgebras are behaviourally equivalent iff
they satisfy the same $\Lambda$-formulae, provided that~$T$ is
finitary (which corresponds to the usual assumption of finite
branching) and~$\Lambda$ is \emph{separating}, i.e.\ for any set~$X$,
every $t\in TX$ is uniquely determined (within~$TX$) by the set
\begin{equation*}
  \{(\lambda,Y_1,\dots,Y_n)\mid \lambda\in\Lambda\text{ $n$-ary}, Y_1,\dots,Y_n\in 2^X,
  t\in\lambda(Y_1,\dots,Y_n)\}.
\end{equation*}
For finitary~$T$, a separating set of modalities always
exists~\cite{Schroder08}.

\section{Fuzzy Relations and Lax Extensions}\label{sec:fuzzy-rel}
\noindent We next introduce the central notion of the paper,
concerning extensions of \emph{fuzzy} (or \emph{real-valued})
relations along a \emph{set functor~$T$, which we fix for the
  remainder of the paper}. We begin by fixing basic concepts and
notation on fuzzy relations. Hemimetrics can be viewed as particular
fuzzy relations, forming a quantitative analogue of preorders;
correspondingly, pseudometrics may be seen as a quantitative analogue
of equivalence relations.

\begin{defn}
  Let $A$ and $B$ be sets. A \emph{fuzzy relation} between $A$ and $B$
  is a map $R\colon A\times B\to[0,1]$, also written
  $R\colon\frel{A}{B}$. We say that~$R$ is \emph{crisp} if
  $R(a,b) \in \{0,1\}$ for all $a\in A,b\in B$ (and generally apply
  the term crisp to concepts that live in the standard two-valued
  setting). The \emph{converse} relation $\rev{R}\colon\frel{B}{A}$ is
  given by $\rev{R}(b,a)=R(a,b)$. For $R,S\colon\frel{A}{B}$, we write
  $R\le S$ if $R(a,b)\le S(a,b)$ for all $a\in A,b\in B$.
\end{defn}

\begin{conv}\label{conv:zero}
  Crisp relations are just ordinary relations.  However, since we are
  working in a metric setting, it will be more natural to use
  the convention that elements $a\in A,b\in B$ are related by a crisp
  relation $R$ if $R(a,b) = 0$, in which case we write $aRb$.
\end{conv}
\begin{conv}[Composition]\label{conv:comp}
  We write composition of fuzzy relations diagrammatically, using~`;'.
  Explicitly, the composite $R_1;R_2\,\colon\frel{A}{C}$ of
  $R_1\colon\frel{A}{B}$ and $R_2\colon\frel{B}{C}$ is defined by
  \begin{equation*}\textstyle
    (R_1;R_2)(a,c) = \inf_{b\in B}R_1(a,b)\oplus R_2(b,c),
  \end{equation*}
  where $\oplus$ denotes \L{}ukasiewicz disjunction:
  $x\oplus y = \min(x+y,1)$. Note that given our previous convention
  on crisp relations, the restriction of this composition operator to
  crisp relations is precisely the standard relational composition. We
  reserve the applicative composition operator~$\circ$ for composition
  of functions. In particular,~$R\colon\frel{A}{B}$ is viewed as a
  function $A\times B\to [0,1]$ whenever~$\circ$ is applied to~$R$.
  Throughout the paper, we use the fact that composition is monotone,
  that is, for $R_1\le R'_1$ and $R_2\le R'_2$ we have $R_1;R_2\le R'_1;R'_2$.
\end{conv}

\begin{defn}[Functions as relations]\label{def:gr}
  The \emph{$\epsilon$-graph} of a function $f\colon A\to B$ is the
  fuzzy relation $\gr{\epsilon,f}\colon\frel{A}{B}$ given by
  $\gr{\epsilon,f}(a,b) = \epsilon$ if $f(a) = b$, and
  $\gr{\epsilon,f}(a,b) = 1$ otherwise.  The $\epsilon$-graph of the
  identity function $\id_A$ is also called the
  \emph{$\epsilon$-diagonal} of~$A$, and denoted by
  $\Delta_{\epsilon,A}$. We refer to $\gr{0,f}$ simply as the
  \emph{graph} of~$f$, also denoted~$\gr{f}$, and to $\Delta_{0,A}$ as
  the \emph{diagonal} of~$A$, which we continue to denote
  as~$\Delta_A$.
\end{defn}

\noindent The following is now straightforward.
\begin{lem}\label{lem:comp-graph}\hfill
  \begin{enumerate}
    \item \label{item:comp-graph-diag} For every function $f\colon
      A\to B$, we have $\Delta_B \le \rev{\gr{f}};\gr{f}$ and
      $\gr{f};\rev{\gr{f}} \le \Delta_A$.
    \item \label{item:comp-graph-left-right} For every
      $R\colon\frel{A'}{B'}$, $f\colon A\to A'$ and $g\colon B\to B'$,
      we have $R \circ (f\times g) = \gr{f}; R; \rev{\gr{g}}$.
  \end{enumerate}
\end{lem}

\noindent Using the notation assembled, we can rephrase the definition
of hemimetric and pseudometric as follows.
\begin{lem}\label{lem:hemimetric-rephrased} Let $d\colon\frel{X}{X}$ be a fuzzy relation.
  \begin{enumerate}
    \item $d$ is a hemimetric iff $d \le \Delta_X$ (reflexivity) and $d
      \le d;d$ (triangle inequality).
    \item $d$ is a pseudometric iff it is a hemimetric and
    additionally $\rev{d} = d$ (symmetry).
  \end{enumerate}
\end{lem}

\noindent We now introduce our central notion of nonexpansive lax extension:
\begin{defn}[Fuzzy relation liftings and lax extensions]
  A \emph{(fuzzy) relation lifting}~$L$ of~$T$ maps each fuzzy
  relation $R\colon\frel{A}{B}$ to a fuzzy relation
  $LR\colon\frel{TA}{TB}$.
  \begin{enumerate}
  \item We say that $L$ \emph{preserves converse} if for all $R$ we
    have
      \begin{align*}
        \text{(L0)} &\quad L(\rev{R}) = \rev{(LR)}.
      \intertext{
        \item We say that~$L$ is a \emph{(fuzzy) lax extension} if it satisfies
      }
        \text{(L1)} &\quad R_1 \le R_2 \Rightarrow LR_1 \le LR_2 \\
        \text{(L2)} &\quad L(R;S) \le LR;LS \\
        \text{(L3)} &\quad L\gr{f} \le \gr{Tf} \text{ and } L(\rev{\gr{f}}) \le \rev{\gr{Tf}}
      \intertext{
        for all sets $A,B$, and $R,R_1,R_2\colon\frel{A}{B}$,
        $S\colon\frel{B}{C}$, $f\colon A\to B$.
    \item A fuzzy lax extension~$L$ is \emph{nonexpansive}, and then
      briefly called a \emph{nonexpansive lax extension}, if
      }
        \text{(L4)} & \quad L\Delta_{\epsilon,A} \le \Delta_{\epsilon,TA}
      \end{align*}
      for all sets~$A$ and $\epsilon>0$.
    \end{enumerate}
\end{defn}
\noindent Axioms (L0)--(L3) are straightforward quantitative
generalizations of the axiomatization of two-valued lax
extensions~\cite{MartiVenema15}; fuzzy lax extensions in this sense
have also been called
\emph{$[0,1]$-relators}~\cite{Gavazzo18,HofmannEA14} (in the more
general setting of quantale-valued relations). Compared
to~\cite{ws:fuzzy-lax}, we do not require fuzzy lax extensions to
satisfy Axiom~(L0) in general; examples of this will be shown in Example~\ref{expl:hausdorff}. This necessitates the addition of the
second clause in Axiom~(L3) (which of course is implied by the first
clause in presence of~(L0)). Axiom~(L4) has no two-valued analogue;
its role and the terminology are explained by
Lemma~\ref{lem:L-nonexpansive} below.

The axioms (L1)--(L3) imply the following basic
property~\cite[Corollary III.1.4.4]{HofmannEA14}:

\begin{lem}[Naturality]\label{lem:transform}
  Let $L$ be a fuzzy lax extension of~$T$, let $R\colon\frel{A'}{B'}$
  be a fuzzy relation, and let $f\colon A\to A', g\colon B\to
  B'$. Then
  \begin{equation*}
    L (R\circ (f\times g)) = LR \circ (Tf\times Tg).
  \end{equation*}
\end{lem}

\begin{proof}
  We need to show two inequalities. For `$\le$', we have
  \begin{align*}
    L(R\circ (f\times g)) 
    &= L (\gr{f};R;\rev{\gr{g}}) &&\myby{Lemma~\ref{lem:comp-graph}.\ref{item:comp-graph-left-right}}\\
    &\le L\gr{f}; LR; L(\rev{\gr{g}}) &&\myby{L2}\\
    &\le \gr{Tf}; LR; \rev{\gr{Tg}} &&\myby{L3}\\
    &= LR \circ (Tf\times Tg). &&\myby{Lemma~\ref{lem:comp-graph}.\ref{item:comp-graph-left-right}}
  \end{align*}
  For `$\ge$', we have
  \begin{align*}
    LR \circ (Tf\times Tg)
    &= \gr{Tf}; LR; \rev{\gr{Tg}} &&\myby{Lemma~\ref{lem:comp-graph}.\ref{item:comp-graph-left-right}}\\
    &= \gr{Tf}; L(\Delta_{A'};R;\Delta_{B'}); \rev{\gr{Tg}} &&\myby{$\Delta$ neutral for $;$}\\
    &\le \gr{Tf}; L(\rev{\gr{f}};\gr{f};R;\rev{\gr{g}};\gr{g}); \rev{\gr{Tg}} &&\mybys{Lemma~\ref{lem:comp-graph}.\ref{item:comp-graph-diag}}{L1}\\
    &\le \gr{Tf}; \rev{\gr{Tf}}; L(\gr{f};R;\rev{\gr{g}}); \gr{Tg}; \rev{\gr{Tg}} &&\mybys{L2}{L3}\\
    &\le L(\gr{f};R;\rev{\gr{g}}) &&\myby{Lemma~\ref{lem:comp-graph}.\ref{item:comp-graph-left-right}}\\
    &\le L(R\circ (f\times g)). &&\myby{Lemma~\ref{lem:comp-graph}.\ref{item:comp-graph-diag}}
    \tag*{\qedhere}
  \end{align*}
\end{proof}

\noindent Using Lemma~\ref{lem:transform}, we can prove the following
characterization of Axiom (L4), which is an important prerequisite for the
Hennessy-Milner theorem.

\begin{lem}\label{lem:L-nonexpansive}
  Let $L$ be a fuzzy lax extension of~$T$. Then the following are
  equivalent:
  \begin{enumerate}
  \item\label{item:L4} $L$ satisfies Axiom~(L4) (i.e.\ is nonexpansive).
  \item\label{item:L4-gr} For all functions $f\colon A\to B$ and all $\epsilon>0$,
    \begin{math}
      L\gr{\epsilon,f}\le\gr{\epsilon,Tf}.
    \end{math}
  \item\label{item:nonexp} For all sets $A,B$, the map $R\mapsto LR$
    is nonexpansive w.r.t.\ the~supremum~metric on~$\frel{A}{B}$.
  \end{enumerate}
\end{lem}
\begin{proof}\hfill

  \begin{itemize}
    \item \emph{\itemref{item:L4}$\iff$\itemref{item:L4-gr}:}
      The implication`$\Leftarrow$' is trivial; we prove `$\Rightarrow$'. We have
      \begin{align*}
        L\gr{\epsilon,f}
          & = L(\Delta_{\epsilon,B}\circ(f\times\id_B)) &&\myby{Definition~\ref{def:gr}}\\
          & = L\Delta_{\epsilon,B}\circ(Tf\times\id_{TB}) && \myby{Lemma~\ref{lem:transform}}\\
          & \le \Delta_{\epsilon,TB}\circ(Tf\times\id_{TB})&& \myby{\ref{item:L4}}\\
          & = \gr{\epsilon,Tf}. &&\myby{Definition~\ref{def:gr}}
      \end{align*}
    \item \emph{\itemref{item:L4}$\implies$\itemref{item:nonexp}:} Let
      $R_1,R_2\colon\frel{A}{B}$ and $\epsilon>0$ such that
      $\supnorm{R_1-R_2}\le\epsilon$; we need to show that
      $\supnorm{LR_1-LR_2}\le\epsilon$. The assumption implies $R_1\le
      R_2;\Delta_{\epsilon,B}$, hence, using (L1), (L2), and \itemref{item:L4},
      \begin{equation*}
        LR_1\le L(R_2;\Delta_{\epsilon,B})
          \le LR_2;L\Delta_{\epsilon,B}
          \le LR_2;\Delta_{\epsilon,TB}.
      \end{equation*}
      Symmetrically, we show $LR_2\le LR_1;\Delta_{\epsilon,TB}$, so
      that $\supnorm{LR_1-LR_2}\le\epsilon$.
  \item\begin{samepage} \emph{\itemref{item:nonexp}$\implies$\itemref{item:L4}:} We
      have $\supnorm{\Delta_{\epsilon,A}-\Delta_A}=\epsilon$, and
      hence by assumption
      $\supnorm{L\Delta_{\epsilon,A}-L\Delta_A}\le\epsilon$. In
      particular,
      $L\Delta_{\epsilon,A}\le L\Delta_A;\Delta_{\epsilon,TA}$, so
      \begin{equation*}
        L\Delta_{\epsilon,A}\le
        L\Delta_A;\Delta_{\epsilon,TA}\le\Delta_{TA};\Delta_{\epsilon,TA}=\Delta_{\epsilon,TA}
      \end{equation*}
      using (L3).\qedhere
  \end{samepage}
  \end{itemize}
\end{proof}

\begin{rem}
  As stated in the introduction, nonexpansiveness of lax extensions
  relates to conditions on certain lax extensions for strong monads
  introduced by Gavazzo~\cite{Gavazzo18}, called \emph{L-continuous
  $V$-relators} (for $V$ a quantale).
  Specifically, as $T$ is a set functor, it has a \emph{tensorial
  strength} $\strength{A}{B}\colon A\times TB\to T(A\times B)$ given
  by $\strength{A}{B}(a,t) = T(b\mapsto(a,b))(t)$. Instantiating to
  the unit interval and using our notation, the axioms of an
  L-continuous $[0,1]$-relator $L$ require that strength is
  nonexpansive, i.e.~that for all sets $A,B,X$ and $Y$ and all fuzzy
  relations $R\colon\frel{A}{B}$ and $S\colon\frel{X}{Y}$ we have
  \begin{equation}\label{eqn:strength}
    L(R\oplus S)\circ (\strength{A}{X}\times\strength{B}{Y}) \le R\oplus LS,
  \end{equation}
  where $\oplus$ is taken pointwise.
  We say that $L$ is \emph{strong} if it satisfies \eqref{eqn:strength}
  and show that $L$ is strong iff it is nonexpansive:
  
  Let $a\in A$, $b\in B$ and let $p\colon X\to A\times X$ and $q\colon Y\to B\times Y$
  be the maps $x\mapsto(a,x)$ and $y\mapsto(b,y)$, respectively.
  Then, unfolding the definition of $\mathsf{str}$ and applying Lemma~\ref{lem:transform},
  we have, for $t_1\in TX$ and $t_2\in TY$,
  \begin{equation*}
    L(R\oplus S)(\strength{A}{X}(a,t_1), \strength{B}{Y}(b,t_2))
    = L(R\oplus S)(Tp(t_1), Tq(t_2))
    = L((R\oplus S)\circ(p\times q))(t_1,t_2).
  \end{equation*}
  Put $\epsilon := R(a,b)$. Then we have, for all $x\in X$ and $y\in Y$,
  \begin{equation*}
    ((R\oplus S)\circ(p\times q))(x,y)
    = R(a,b)\oplus S(x,y) = (\Delta_{\epsilon,X};S)(x,y),
  \end{equation*}
  so that $(R\oplus S)\circ(p\times q) = \Delta_{\epsilon,X};LS$.
  Similarly, for $a$ and $b$ fixed like this we have that
  $(R\oplus LS)((a,t_1),(b,t_2)) = (\Delta_{\epsilon,TX};LS)(t_1,t_2)$.
  Thus, \eqref{eqn:strength} is equivalent to the requirement that 
  \begin{equation}\label{eqn:strength2}
    L(\Delta_{\epsilon,X};S) \le \Delta_{\epsilon,TX};LS
  \end{equation}
  for all sets $X,Y$, all $S\colon\frel{X}{Y}$ and all $\epsilon\ge 0$.
  Finally, we show that \eqref{eqn:strength2} is equivalent to (L4):
  \begin{itemize}
    \item `$\Rightarrow$': We have
      \begin{math}
        L(\Delta_{\epsilon,X})
        = L(\Delta_{\epsilon,X};\Delta_X)
        \le \Delta_{\epsilon,TX};L\Delta_X
        \le \Delta_{\epsilon,TX};\Delta_{TX}
        = \Delta_{\epsilon,TX},
      \end{math}
      using \eqref{eqn:strength2} and (L3) in the inequalities
      and neutrality of $\Delta$ in the equalities.
    \item `$\Leftarrow$': By (L2) and (L4), we have
      \begin{math}
        L(\Delta_{\epsilon,X};S)
        \le L\Delta_{\epsilon,X};LS
        \le \Delta_{\epsilon,TX};LS.
      \end{math}
  \end{itemize}
\end{rem}

\noindent As indicated previously, many existing approaches to behavioural
metrics
(e.g.~\cite{bw:behavioural-pseudometric,bbkk:coalgebraic-behavioral-metrics})
are based on lifting functors to pseudometric spaces. In the present
framework, every lax extension induces a functor lifting to hemimetric
spaces; or to pseudometric spaces if the lax extension preserves converse:

\begin{lem}\label{lem:functor-lifting} Let $L$ be a fuzzy lax extension.
  \begin{enumerate}
    \item Let $d\colon\frel{X}{X}$ be a hemimetric. Then $Ld$ is a
      hemimetric on $TX$. If $d$ is a pseudometric and $L$ preserves
      converse, then $Ld$ is a pseudometric as well.
    \item For every nonexpansive map $f\colon (X,d_1)\to (Y,d_2)$ of
      hemimetric spaces, the map\\ $Tf\colon(TX,Ld_1)\to (TY,Ld_2)$ is
      nonexpansive.
  \end{enumerate}
\end{lem}

\begin{proof}\hfill
  \begin{enumerate}
    \item Using Lemma~\ref{lem:hemimetric-rephrased} and the laws of lax
      extensions, we have $Ld \le L\Delta_X \le \Delta_{TX}$
      and $Ld \le L(d;d) \le Ld;Ld$, so $Ld$ is a hemimetric.
      If $L$ preserves converse and $d$ is a pseudometric, then $\rev{(Ld)}
      = L(\rev{d}) = Ld$, so that $Ld$ is a pseudometric.
    \item Let $f\colon(X,d_1)\to(Y,d_2)$ be a nonexpansive map, that is
      \begin{math}
        d_2\circ (f\times f) \le d_1.
      \end{math}
      Then $Tf$ is nonexpansive as well, by naturality
      (Lemma~\ref{lem:transform}) and monotonicity:
      \begin{equation*}
        Ld_2\circ(Tf\times Tf) = L(d_2\circ (f\times f)) \le Ld_1. \qedhere
      \end{equation*}
  \end{enumerate}
\end{proof}

\noindent As a consequence of Lemma~\ref{lem:functor-lifting}, every
fuzzy lax extension of $T\colon\Set\to\Set$ gives rise to a functor
\begin{math}
  \overline{T}\colon\HMet\to\HMet
\end{math}
on the category $\HMet$ of hemimetric spaces and nonexpansive maps
that \emph{lifts}~$T$ in the sense that $U\circ\overline T=T\circ U$,
where $U\colon \HMet\to\Set$ is the functor that forgets the
hemimetric. Similarly, every converse-preserving fuzzy lax extension
induces a lifting of~$T$ to the category of pseudometric spaces.

Much of the development will be based on finitary functors; for
instance, we need a finitary functor so we can give an explicit syntax
for the characterizing logic of a lax extension. We can capture a
broader class of functors, specifically those functors that are suitably
approximated by their finitary parts in the sense that the finitary part
forms a dense subset of the unrestricted functor.

\begin{defn}\label{defn:density}
  Let $(X,d)$ be a hemimetric space. A set $A\subseteq X$ is \emph{dense}
  if for all $x\in X$ and all $\epsilon > 0$ there exists some $a\in A$
  such that both $d(x,a)\le\epsilon$ and $d(a,x)\le\epsilon$.
\end{defn}

\noindent This notion of density for hemimetrics coincides with an
existing one for quantale-valued distances~\cite{FlaggKopperman97}. In
particular, it is essential to require both inequalities in
Definition~\ref{defn:density}, as otherwise certain pathological cases
of dense subsets may occur. For instance, if we left out the second
inequality from the above definition, then the singleton set $\{1\}$
would be a dense subset of the unit interval $[0,1]$ under the
hemimetric $d_\ominus$ (Section~\ref{sec:prelim}).

Equipped with this definition of density, we proceed to introduce the
following condition which allows for the treatment of lax extensions
of certain non-finitary functors.

\begin{defn}
  A fuzzy lax extension $L$ of the functor $T$ is \emph{finitarily
  separable} if for every set $X$, $T_\omega X$ is a dense subset
  of $TX$ w.r.t. the hemimetric $L\Delta_X$.
\end{defn}

\noindent Clearly, any lax extension of a finitary functor is
finitarily separable. The prototypical example of a finitarily
separable lax extension of a non-finitary functor is the Kantorovich
lifting of the discrete distribution functor $\dfun$
(Example~\ref{expl:kantorovich}.\ref{item:prob-kantorovich}); that is,
every discrete distribution can be approximated, under the usual
Kantorovich metric, by finitely supported distributions.

We conclude the section with a basic example of a nonexpansive lax
extension, deferring further examples to the sections on systematic
constructions of such extensions (Sections~\ref{sec:kantorovich}
and~\ref{sec:wasserstein}):
\begin{expl}[Hausdorff lifting]\label{expl:hausdorff}
  The \emph{Hausdorff lifting} is the relation lifting $H$ for the
  powerset functor $\Pow$, defined for fuzzy relations
  $R\colon\frel{A}{B}$ by
  \begin{align*}
    HR(U,V) = \max(\adjustlimits\sup_{a\in U}\inf_{b\in V} R(a,b),\;
                   \adjustlimits\sup_{b\in V}\inf_{a\in U} R(a,b))
  \end{align*}
  for $U\subseteq A, V\subseteq B$.  The Hausdorff lifting can be
  viewed as a quantitative analogue of the Egli-Milner extension
  (Section~\ref{sec:prelim}), where $\sup$ replaces universal
  quantification and $\inf$ replaces existential quantification. It is
  shown already in~\cite{HofmannEA14} that~$H$ is a fuzzy lax
  extension. Indeed, it is easy to see that~$H$ is also
  converse-preserving and nonexpansive. These properties will also
  follow from the results of Section~\ref{sec:wasserstein}, where we
  show that~$H$ is in fact an instance of the Wasserstein lifting.
  $H$ is not finitarily separable, because for every set $X$ we have
  $H\Delta_X = \Delta_{\Pow X}$.

  We may also consider asymmetric versions of the Hausdorff lifting
  by simply omitting one of the two terms in the definition, putting
  \begin{equation*}
    H^\leftarrow R(U,V) = \adjustlimits\sup_{a\in U}\inf_{b\in V} R(a,b)
    \qquad\text{and}\qquad
    H^\rightarrow R(U,V) = \adjustlimits\sup_{b\in V}\inf_{a\in U} R(a,b)
  \end{equation*}
  for $U\subseteq A, V\subseteq B$. Both $H^\leftarrow$ and $H^\rightarrow$
  are nonexpansive fuzzy lax extensions, but neither of them preserves converse.
\end{expl}

\section{Quantitative Simulations}\label{sec:bisim}

\noindent We next identify a notion of simulation based on a lax
extension~$L$ of the functor~$T$; similar concepts appear in work on
quantitative applicative bisimilarity~\cite{Gavazzo18}. We define
behavioural distance based on this notion, and show coincidence with
the distance defined via the pseudometric lifting induced by~$L$
according to Lemma~\ref{lem:functor-lifting}.
\begin{defn}\label{defn:l-dist}
  Let $L$ be a lax extension of $T$, and let
  $\alpha\colon A\to TA$ and $\beta\colon B\to TB$ be coalgebras.
  \begin{enumerate}
    \item A fuzzy relation $R\colon\frel{A}{B}$ is an
      \emph{$L$-simulation} if
        \begin{math}
        LR \circ (\alpha\times\beta) \le R.
      \end{math}
    \item $R$ is an \emph{$L$-bisimulation} if both $R$ and $\rev{R}$
      are $L$-simulations.
    \item We define \emph{$L$-behavioural distance}
      $d^L_{\alpha,\beta}\colon\frel{A}{B}$ to be the infimum of all $L$-simulations:
      \begin{equation*}
        d^L_{\alpha,\beta} = \inf \{ R\colon\frel{A}{B} \mid R \text{ is an $L$-simulation} \}.
      \end{equation*}
      If $\alpha=\beta$, we write $d^L_\alpha = d^L_{\alpha,\beta}$ instead.
      \label{item:l-dist}
  \end{enumerate}
\end{defn}

\begin{rem}
  Putting Definition~\ref{defn:l-dist} in other words, an $L$-simulation is
  precisely a prefix point for the map
  \begin{math}
    F(R) = LR \circ (\alpha\times\beta).
  \end{math}
  Note that $F$ is monotone by (L1).
  This means that, according to the Knaster-Tarski fixpoint theorem,
  $d^L_{\alpha,\beta}$ is itself a prefix point (i.e. an $L$-simulation),
  and also the least fixpoint of $F$, i.e.\
  \begin{math}
    d^L_{\alpha,\beta} = L d^L_{\alpha,\beta} \circ (\alpha\times\beta).
  \end{math}
  In particular, the infimum in Definition~\ref{defn:l-dist}.\ref{item:l-dist}
  is always a minimum.
\end{rem}

\begin{expl}\label{expl:metric-ts}
  The weighted transition systems discussed in the introduction can be
  modelled as coalgebras for the functor $\Pow_\omega(M\times-)$, and the
  simulation distance given there is then $L$-behavioural distance for
  the fuzzy lax extension $L$ defined for $R\colon\frel{A}{B}$ by
  \begin{equation*}
    LR(U,V) = \adjustlimits\sup_{(m,a)\in U}\inf_{(n,b)\in V} d_M(m,n) + \lambda R(a,b),
  \end{equation*}
  where $U\subseteq M\times A, V\subseteq M\times B$. To ensure that
  all values of $LR$ lie in the unit interval $[0,1]$, we require that
  $d_M(m,n) \le 1-\lambda$ for all $m,n\in M$. If $M$ is finite (as is
  the case in~\cite{LarsenEA11}) this can always be achieved by rescaling.
\end{expl}

\noindent We note the following facts about $L$-simulations:

\begin{lem}\label{lem:l-sim-properties} Let $L$ be a fuzzy lax
  extension, and let $\alpha\colon A\to TA$, $\beta\colon B\to TB$ and
  $\gamma\colon C\to TC$ be coalgebras. Then
  \begin{enumerate}
    \item $\Delta_A$ is an $L$-simulation.\label{item:delta-bisim}
    \item For any $L$-simulations $R\colon\frel{A}{B}$ and 
      $S\colon\frel{B}{C}$, $R;S$ is an $L$-simulation.\label{item:comp-bisim}
  \end{enumerate}
\end{lem}
\begin{proof}
  For Item~\itemref{item:delta-bisim}, we have
  \begin{equation*}
    L\Delta_A\circ(\alpha\times\alpha) \le
    \Delta_{TA}\circ(\alpha\times\alpha) =
    \rev{\gr{\alpha}};\Delta_{TA};\gr{\alpha} =
    \rev{\gr{\alpha}};\gr{\alpha} \le \Delta_A
  \end{equation*}
  by (L3) and both parts of Lemma~\ref{lem:comp-graph}.
  For Item~\itemref{item:comp-bisim}, we compute
  \begin{align*}
    L(R;S) \circ (\alpha\times\gamma) \hspace{-2cm} &\\
    &\le (LR;LS) \circ (\alpha\times\gamma) &&\myby{L2} \\
    &= \gr{\alpha}; LR; LS; \rev{\gr{\gamma}} &&\myby{Lemma~\ref{lem:comp-graph}.\ref{item:comp-graph-left-right}}\\
                                                        &\le \gr{\alpha}; LR; \rev{\gr{\beta}}; \gr{\beta}; LS; \rev{\gr{\gamma}} &&
    \myby{Lemma~\ref{lem:comp-graph}.\ref{item:comp-graph-diag}}\\
    &= LR \circ (\alpha\times\beta); LS \circ (\beta\times\gamma) &&\myby{Lemma~\ref{lem:comp-graph}.\ref{item:comp-graph-left-right}}\\
    &\le R;S. &&\myby{assumption}\hfill\qedhere
  \end{align*}
\end{proof}

\noindent For converse-preserving lax extensions, this notion of
simulation is actually one of bisimulation, more precisely:

\begin{lem}\label{lem:l-bisim}
  If $L$ preserves converse, then every $L$-simulation is an $L$-bisimulation.
\end{lem}
\begin{proof}
  Let $\alpha\colon A\to TA$ and $\beta\colon B\to TB$ be coalgebras
  and let $R$ be an $L$-simulation. Then by (L0) we have
  \begin{equation*}
    L(\rev{R})\circ(\beta\times\alpha) = \rev{(LR)}\circ(\beta\times\alpha)
     = \rev{(LR\circ(\alpha\times\beta))} \le \rev{R}. \qedhere
  \end{equation*}
\end{proof}

\noindent As $L$-behavioural distance is the least
$L$-simulation, we have

\begin{lem}\label{lem:l-dist-hemimetric}
  For every coalgebra $\alpha\colon A\to TA$, $d^L_\alpha$ is a hemimetric.
  If $L$ preserves converse, then $d^L_\alpha$ is a pseudometric.
\end{lem}
\begin{proof}
  Since $d^L_\alpha$ is an $L$-simulation, both $\Delta_A$ and
  $d^L_\alpha;d^L_\alpha$ are $L$-simulations by
  Lemma~\ref{lem:l-sim-properties}. As $d^L_\alpha$ is the
  \emph{least} $L$-simulation, this implies $d^L_\alpha\le\Delta_A$ and
  $d^L_\alpha \le d^L_\alpha; d^L_\alpha$, so that $d^L_\alpha$ is a
  hemimetric by Lemma~\ref{lem:hemimetric-rephrased}.

  In the converse-preserving case, we additionally have that
  $\rev{(d^L_\alpha)}$ is an $L$-simulation by
  Lemma~\ref{lem:l-bisim}, and therefore
  $d^L_\alpha\le\rev{(d^L_\alpha)}$, making $d^L_\alpha$ a
  pseudometric by Lemma~\ref{lem:hemimetric-rephrased}.
\end{proof}

\begin{rem}\label{rem:metric-distance}
  As announced above and as we show next, existing generic notions of
  behavioural distance defined via functor
  liftings~\cite{bbkk:coalgebraic-behavioral-metrics} agree with the
  one given above (when both apply). Specifically, when applied to the
  functor lifting induced by a converse-preserving lax extension~$L$
  of~$T$ according to Lemma~\ref{lem:functor-lifting}, the definition
  of behavioural distance via functor liftings amounts to taking the
  same least fixpoint as in Definition~\ref{defn:l-dist} but only over
  pseudometrics instead of over fuzzy relations~\cite[Lemma~6.1]{bbkk:coalgebraic-behavioral-metrics}.  Now let $(A,\alpha)$
  be a coalgebra and denote the behavioural distance on~$A$ according
  to the definition in~\cite{bbkk:coalgebraic-behavioral-metrics}
  by~$\bar d_\alpha$. We claim that $\bar d_\alpha=d^L_\alpha$,
  with~$d^L_\alpha$ defined according to
  Definition~\ref{defn:l-dist}. Indeed, `$\ge$' is trivial since
  $\bar d_\alpha$ is, by definition, an $L$-bisimulation, and `$\le$'
  is immediate from~$d^L_\alpha$ being a pseudometric
  (Lemma~\ref{lem:l-dist-hemimetric}).
\end{rem}

\begin{rem}\label{rem:diagonals}
  Every converse-preserving fuzzy lax extension~$L$ induces a crisp
  lax extension~$L_c$, where for any crisp relation $R$,
  $L_c R = (LR)^{-1}[\{0\}] \subseteq TA\times TB$ (recall
  Convention~\ref{conv:zero}). It is easily checked that~$L_c$
  preserves diagonals (Section~\ref{sec:prelim}) iff
  \begin{equation}\label{eq:diag-metric}
    \text{$L\Delta_A$ is a metric for each set~$A$.}
  \end{equation}
  By results on lax extensions cited in Section~\ref{sec:prelim},
  $L_c$-bisimilarity coincides with behavioural equivalence in this
  case, i.e.\ if~$L$ satisfies~\eqref{eq:diag-metric}, then $L$
  characterizes behavioural equivalence: Two states $a\in A$ and $b\in B$
  in coalgebras $(A,\alpha)$ and $(B,\beta)$ are behaviourally equivalent
  iff $d^L_{\alpha,\beta}(a,b) =~0$.
\end{rem}

\begin{expl}[Small bisimulations]
  \label{expl:small-bisim}
  We give an example of a bisimulation for a lax extension of the
  functor $TX = [0,1]\times\Pow X$.  Coalgebras for $T$ are Kripke
  frames where each state is labelled with a number from the unit
  interval. They are similar to the weighted transition systems
  from~\cite{LarsenEA11}, except that here the labels are on the
  states rather than on the transitions. This $T$ has a
  converse-preserving nonexpansive lax extension $L$, defined for
  fuzzy relations $R\colon\frel{A}{B}$ by
  \begin{equation*}
    LR((p,U),(q,V)) = \textstyle\frac{1}{2}(|p-q| + HR(U,V)),
  \end{equation*}
  where $p,q\in[0,1], U\subseteq A, V\subseteq B$, and $H$ is the
  Hausdorff lifting (Example~\ref{expl:hausdorff}). The 
  idea behind this definition is that the $L$-behavioural distance of
  two states is the supremum of the accumulated branching-time
  differences between state labels over all runs of a process starting
  at these states. The factor $\frac{1}{2}$ ensures that the total
  distance is at most~$1$ by discounting the differences at later
  stages with exponentially decreasing factors.

  Now consider the $T$-coalgebras $(A,\alpha)$ and $(B,\beta)$ below:
  
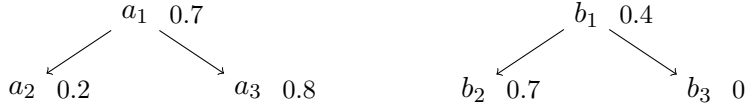
\begin{equation*}
    \begin{tikzpicture}
      \node[label=right:{\small 0.7}] (a1) at (1.5,1) {$a_1$};
      \node[label=right:{\small 0.2}] (a2) at (0,0)   {$a_2$};
      \node[label=right:{\small 0.8}] (a3) at (3,0)   {$a_3$};
      \node[label=right:{\small 0.4}] (b1) at (7.5,1) {$b_1$};
      \node[label=right:{\small 0.7}] (b2) at (6,0)   {$b_2$};
      \node[label=right:{\small 0}]   (b3) at (9,0)   {$b_3$};

      \draw[->] (a1) -- (a2);
      \draw[->] (a1) -- (a3);
      \draw[->] (b1) -- (b2);
      \draw[->] (b1) -- (b3);
    \end{tikzpicture}
  \end{equation*}
  \noindent We put
  $R(a_1,b_1) = 0.2, R(a_2,b_3) = R(a_3,b_2) = 0.1$, and
  $R(a_i,b_j) = 1$ in all other cases. We show that $R$ is an $L$-bisimulation
  witnessing that $d^L_{\alpha,\beta}(a_1,b_1) \le 0.2$, even though it is
  clearly neither reflexive nor symmetric on the disjoint union of the
  systems (it is easy to come up with similar but slightly larger
  examples where~$R$ also fails to be transitive, i.e.~to satisfy the
  triangle inequality).
  
  Specifically, we need to show for each $a_i$ and $b_j$ that
  $LR(\alpha(a_i),\beta(b_j)) \le R(a_i,b_j)$. The cases with
  $R(a_i,b_j) = 1$ are trivial; in the other cases we have:
  \begin{align*}
    HR(\{a_2,a_3\},\{b_2,b_3\}) &= \max(\adjustlimits\max_{a\in \{a_2,a_3\}}\min_{b\in \{b_2,b_3\}} R(a,b),\;
    \adjustlimits\max_{b\in \{b_2,b_3\}}\min_{a\in \{a_2,a_3\}} R(a,b)) \\
    &= \max(0.1,0.1) = 0.1 \\
    LR(\alpha(a_1),\beta(b_1))
    &= LR((0.7,\{a_2,a_3\}),(0.4,\{b_2,b_3\})) \\
    &= {\textstyle\frac{1}{2}} (|0.7-0.4| + HR(\{a_2,a_3\},\{b_2,b_3\})) \\
    &= \textstyle\frac{1}{2} (0.3 + 0.1)
      = 0.2 = R(a_1,b_1) \\
    LR(\alpha(a_2),\beta(b_3))
    &= LR((0.2,\emptyset),(0,\emptyset))
    = \textstyle\frac{1}{2} 0.2 = 0.1 = R(a_2,b_3) \\
    LR(\alpha(a_3),\beta(b_2))
    &= LR((0.8,\emptyset),(0.7,\emptyset))
    = \textstyle\frac{1}{2} 0.1 \le 0.1 = R(a_3,b_2)
  \end{align*}

\end{expl}
\noindent As indicated previously, quantitative Hennessy-Milner
theorems can only be expected to hold for nonexpansive lax
extensions. The key observation is the following. By standard fixpoint
theory, $L$-behavioural distance can be approximated from below by an
ordinal-indexed increasing chain:
\begin{defn}\label{defn:approximant}
  Let~$L$ be a lax extension of~$T$, and let $(A,\alpha)$, $(B,\beta)$ be
  $T$-coalgebras. The sequence of \emph{approximants} of ($L$-behavioural distance)
  $d^L_{\alpha,\beta}$ are the fuzzy relations
  $d_\kappa\colon\frel{A}{B}$, indexed over ordinal numbers~$\kappa$,
   inductively defined by
  \begin{equation*}\textstyle
    d_0 = 0, \quad\quad
    d_{\kappa+1} = Ld_\kappa \circ (\alpha\times\beta),\quad\quad
    d_\lambda = \sup_{\kappa<\lambda} d_\kappa \quad\text{($\lambda$ limit ordinal)}.
  \end{equation*}
\end{defn}
\noindent We show some basic properties of these sequences:
\begin{lem}\label{lem:approximant}
  Let~$L$ be a lax extension of~$T$, let $(A,\alpha)$, $(B,\beta)$ be
  $T$-coalgebras, and let $(d_\kappa\colon\frel{A}{B})_\kappa$ be the
  sequence of approximants of $d^L_{\alpha,\beta}$. Then:
  \begin{enumerate}
    \item The sequence $(d_\kappa)_\kappa$ is increasing. \label{item:approx-inc}
    \item We have $d_\kappa\le d^L_{\alpha,\beta}$ for each ordinal $\kappa$. \label{item:approx-below}
    \item Let $(A',\alpha')$, $(B',\beta')$ be $T$-coalgebras, let $f\colon A\to A',
      g\colon B\to B'$ be coalgebra morphisms, and let
      $(d'_\kappa\colon\frel{A'}{B'})_\kappa$ be the
      sequence of approximants of $d^L_{\alpha',\beta'}$.
      Then $d_\kappa = d'_\kappa\circ(f\times g)$ for each ordinal $\kappa$. \label{item:approx-morph}
  \end{enumerate}
\end{lem}

\begin{proof}\hfill
  \begin{enumerate}
    \item Because $d_\lambda = \sup_{\kappa<\lambda} d_\kappa$ for
      limit ordinals $\lambda$, it is enough to show that $d_\kappa\le
      d_{\kappa+1}$ for each ordinal $\kappa$. We proceed by induction;
      there are three cases. $d_0\le d_1$ holds trivially.
      For successor ordinals $\kappa+1$, we have
      \begin{math}
        d_{\kappa+1} = Ld_\kappa\circ(\alpha\times\beta)
        \le Ld_{\kappa+1}\circ(\alpha\times\beta) = d_{\kappa+2},
      \end{math}
      where we used (L1) and the inductive hypothesis in the
      inequality. Finally, for limit ordinals~$\lambda$, we have
      $Ld_\kappa \le Ld_\lambda$ for each $\kappa<\lambda$ by (L1), so
      that $\sup_{\kappa<\lambda} Ld_\kappa \le Ld_\lambda$. Therefore,
      \begin{math}
      d_\lambda
        = \sup_{\kappa<\lambda} d_\kappa
        \le \sup_{\kappa<\lambda} d_{\kappa+1}
        = \sup_{\kappa<\lambda} Ld_\kappa\circ(\alpha\times\beta)
        \le Ld_\lambda\circ(\alpha\times\beta) = d_{\lambda+1}.
      \end{math}
    \item We proceed by induction; the cases for $0$ and limit
      ordinals are trivial. For successor ordinals, we have
      \begin{math}
        d_{\kappa+1} = Ld_\kappa\circ(\alpha\times\beta)
        \le Ld^L_{\alpha,\beta}\circ(\alpha\times\beta) = d^L_{\alpha,\beta}
      \end{math}
      by the inductive hypothesis, by (L1), and by definition of $d^L_{\alpha,\beta}$.
    \item Again, we proceed by induction; the cases for $0$ and limit
      ordinals are immediate from the definition. For $\kappa+1$ a successor ordinal,
      we compute
      \begin{align*}
        d_{\kappa+1}
        &= Ld_\kappa\circ(\alpha\times\beta) && \myby{definition of $d_{\kappa+1}$} \\
        &= L(d'_\kappa\circ(f\times g))\circ(\alpha\times\beta) && \myby{IH} \\
        &= Ld'_\kappa\circ(Tf\times Tg)\circ(\alpha\times\beta) && \myby{Lemma~\ref{lem:transform}} \\
        &= Ld'_\kappa\circ(\alpha\times\beta)\circ(f\times g) && \myby{$f,g$ morphisms} \\
        &= d'_{\kappa+1}\circ(f\times g). && \myby{definition of $d'_{\kappa+1}$} \qedhere
      \end{align*}
  \end{enumerate}
\end{proof}

\noindent Crucially, if~$L$ is nonexpansive and finitarily separable, then the
chain of approximants stabilizes after~$\omega$ steps. Formally:
\begin{thm}\label{thm:fixpoint}
  Let~$L$ be a nonexpansive finitarily separable lax extension of~$T$.
  Given $T$-coalgebras $(A,\alpha)$, $(B,\beta)$, let $(d_\kappa\colon\frel{A}{B})_\kappa$,
  be the approximants of $d^L_{\alpha,\beta}$.
  Then
  \begin{enumerate}[(i)]
  \item \label{item:omega-converge}$d_{\omega+1} = d_\omega$, and
  \item\label{item:L-dist-omega} $L$-behavioural distance
    $d^L_{\alpha,\beta}$ equals $d_\omega$.
  \end{enumerate}
\end{thm}

\noindent To prove Theorem~\ref{thm:fixpoint} in the case of
non-finitary $T$, we make use of unravellings of coalgebras:

\begin{defn}[Unravelling]\label{defn:unravelling}
  Let $(C,\gamma)$ be a $T$-coalgebra and put $C^+ = \bigcup_{m\ge 1} C^m$.
  \begin{enumerate}
    \item For $\bar c = (c_1,\dots,c_m)\in C^m$ and $c\in C$ we put $\last(\bar
      c) = c_m$ and $\app_{\bar c}(c) = (c_1,\dots,c_m,c)$, defining maps
      $\last\colon C^+\to C$ for each $m\ge 1$ and $\app_{\bar c}\colon C
      \to C^+$ for each $m\ge 1$ and $\bar c\in C^m$.
    \item The \emph{unravelling} of $(C,\gamma)$ is the $T$-coalgebra $(C^+, \gamma^+)$,
      where $\gamma^+\colon C^+ \to TC^+$ is given by
      \begin{equation*}
        \gamma^+(\bar c) = T\app_{\bar c}(\gamma(\last(\bar c)).
      \end{equation*}
  \end{enumerate}
\end{defn}

\noindent Every coalgebra is behaviourally equivalent to its unravelling:

\begin{lem}\label{lem:unravel-morph}
  For every $T$-coalgebra $(C,\gamma)$,
  \begin{enumerate}[label=(\roman*)]
    \item the map $\last\colon (C^+,\gamma^+) \to (C,\gamma)$ 
      is a coalgebra morphism; \label{item:unravel-morphism} and 
    \item every state $c\in C$ is behaviourally equivalent to the state
      $(c) \in C^+$. \label{item:unravel-bisim}
  \end{enumerate}
\end{lem}
\noindent This fact is essentially standard; we give a proof for the
sake of completeness:
\begin{proof}\hfill
  \begin{enumerate}[label=(\roman*)]
    \item Let $\bar c \in C^+$. Then clearly $\last\circ\app_{\bar c}=\id$ by definition
      and therefore
      \begin{equation*}
        (T\last \circ \gamma^+)(\bar c)
        = T\last(T\app_{\bar c}(\gamma(\last(\bar c))))
        = T\id(\gamma(\last(\bar c)))
        = (\gamma\circ\last)(\bar c).
      \end{equation*}
    \item This is immediate from \ref{item:unravel-morphism}, as
      behavioural equivalence is witnessed by the coalgebra morphisms
      $\id\colon C\to C$ and $\last\colon C^+\to C$. \qedhere
  \end{enumerate}
\end{proof}

\begin{proof}[Proof of Theorem~\ref{thm:fixpoint}]
 
  By the fixpoint definition of~$d^L_{\alpha,\beta}$ and
  Lemma~\ref{lem:approximant}.\ref{item:approx-below},~\ref{item:L-dist-omega}
  is immediate from~\ref{item:omega-converge}. We
  prove~\ref{item:omega-converge}, i.e.\ that
  $Ld_{\omega}(\alpha(a),\beta(b))= d_\omega(a,b)$ for all $a\in A$,
  $b\in B$. We begin by assuming that $T$ is finitary, and generalize
  to the non-finitary case later.

  Since~$T$ is finitary, there exist finite subsets
  $A_0\subseteq A$, $B_0\subseteq B$ and $s\in TA_0$, $t\in TB_0$ such
  that $\alpha(a)=Ti(s)$ and $\beta(b)=Tj(t)$, where
  $i\colon A_0\to A$ and $j\colon B_0\to B$ are the inclusion maps. We
  then have
  $Ld_{\omega}(\alpha(a),\beta(b))=L(d_\omega\circ(i\times j))(s,t)$
  by naturality (Lemma~\ref{lem:transform}). By Lemma~\ref{lem:approximant}.\ref{item:approx-inc}, the maps
  $d_n\circ(i\times j)$ converge to $d_\omega\circ(i\times j)$
  pointwise, and therefore also under the supremum metric (i.e.\
  uniformly), since $A_0\times B_0$ is finite. Since $L$ is nonexpansive,
  it is also continuous w.r.t.\ the supremum metric by
  Lemma~\ref{lem:L-nonexpansive}, so it follows that
  \begin{align*}
    Ld_\omega(\alpha(a),\beta(b)) &= L(d_\omega\circ(i\times j))(s,t) &&\myby{naturality}\\
    & =\textstyle\sup_{n<\omega} L(d_n\circ(i\times j))(s,t) &&\myby{$L$ continuous}\\
    & =\textstyle\sup_{n<\omega} Ld_n(\alpha(a),\beta(b)) &&\myby{naturality}\\
    & = \textstyle\sup_{n<\omega} d_{n+1}(a,b)
      = d_\omega(a,b). &&\myby{definition of $d_{n+1}, d_\omega$}
  \end{align*}

  \noindent This covers the finitary case. In the general case, we
  make use of the unravellings $(A^+,\alpha^+)$ and $(B^+,\beta^+)$,
  as well as the sequence $(d^+_\kappa\colon\frel{A^+}{B^+})_\kappa$
  of approximants of $d^L_{\alpha^+,\beta^+}$. We can assume
  w.l.o.g.~that~$A\neq\emptyset\neq B$; then the inclusions
  $A^m\into A^+$, $B^m\into B^+$ (for $m\ge 1$) are preserved by~$T$, and for
  readability we assume in the following that $TA^m$ is in fact a
  subset of $TA^+$; similarly for~$B^m$ and~$T_\omega$, with
  naturality of~$L$ guaranteeing that the identification does not
  affect lifted distance. Now let $\epsilon>0$. As~$L$ is finitarily
  separable, we can construct $T_\omega$-coalgebras
  $\alpha^\epsilon\colon A^+ \to T_\omega A^+$ and
  $\beta^\epsilon\colon B^+ \to T_\omega B^+$ approximating $\alpha^+$
  and $\beta^+$ respectively. Specifically, for every $\bar a\in A^m$
  we have $\alpha^+(\bar a)\in TA^{m+1}$ by definition, and as
  $T_\omega A^{m+1}$ is dense in $TA^{m+1}$, we can choose an element
  $\alpha^\epsilon(\bar a) \in T_\omega A^{m+1}$ such that
  \begin{equation}\label{eqn:approx-coalg-a}
    L\Delta_{A^+}(\alpha^+(\bar a), \alpha^\epsilon(\bar a)) \le \epsilon\cdot 3^{-m}
    \quad\text{and}\quad
    L\Delta_{A^+}(\alpha^\epsilon(\bar a), \alpha^+(\bar a)) \le \epsilon\cdot 3^{-m}.
  \end{equation}
  Similarly, for each $\bar b\in B^m$ we choose $\beta^\epsilon(\bar
  b) \in T_\omega B^{m+1}$ such that
  \begin{equation}\label{eqn:approx-coalg-b}
    L\Delta_{B^+}(\beta^+(\bar b), \beta^\epsilon(\bar b)) \le \epsilon\cdot 3^{-m}
    \quad\text{and}\quad
    L\Delta_{B^+}(\beta^\epsilon(\bar b), \beta^+(\bar b)) \le \epsilon\cdot 3^{-m}.
  \end{equation}
  We denote the sequence of approximants of
  $d^L_{\alpha^\epsilon,\beta^\epsilon}$ by
  $(d^\epsilon_\kappa\colon\frel{A^+}{B^+})_\kappa$ and
  show by induction that the $d^\epsilon_\kappa$ approximate the
  $d^+_\kappa$ in the following sense: for all $m \ge 1$ and all $\bar
  a\in A^m, \bar b\in B^m$,
  \begin{equation}\label{eqn:approx-approximant}
    |d^\epsilon_\kappa(\bar a,\bar b) - d^+_\kappa(\bar a,\bar b)| \le \epsilon\cdot 3^{1-m}
  \end{equation}
  for all ordinals $\kappa$.
  
  For $\kappa=0$ this clearly holds. For the inductive step from
  $\kappa$ to $\kappa+1$, we note again that for $\bar a\in A^m$ and
  $\bar b\in B^m$ we have $\alpha^+(\bar a) \in TA^{m+1}$ and
  $\beta^+(\bar b) \in TB^{m+1}$ by definition . Therefore, by
  Lemma~\ref{lem:L-nonexpansive}.\ref{item:nonexp} and the inductive
  hypothesis, we have
  \begin{equation}\label{eqn:approx-ast-eps}
    |Ld^\epsilon_\kappa(\alpha^+(\bar a),\beta^+(\bar b))
    - Ld^+_\kappa(\alpha^+(\bar a),\beta^+(\bar b))|
    \le \epsilon\cdot 3^{1-(m+1)} = \epsilon\cdot 3^{-m},
  \end{equation}
  so that we compute:
  \begin{align*}
    &d^\epsilon_{\kappa+1}(\bar a,\bar b) &&\\
    &= Ld^\epsilon_\kappa(\alpha^\epsilon(\bar a), \beta^\epsilon(\bar b)) &&\myby{definition of $d^\epsilon_{\kappa+1}$}\\
    &= L(\Delta_{A^+};d^\epsilon_\kappa;\Delta_{B^+})
        (\alpha^\epsilon(\bar a), \beta^\epsilon(\bar b)) &&\myby{$\Delta$ neutral for $;$}\\
    &\le L\Delta_{A^+}(\alpha^\epsilon(\bar a),\alpha^+(\bar a)) +
        Ld^\epsilon_\kappa(\alpha^+(\bar a),\beta^+(\bar b)) +
        L\Delta_{B^+}(\beta^+(\bar b),\beta^\epsilon(\bar b)) &&\myby{L2}\\
    &\le Ld^\epsilon_\kappa(\alpha^+(\bar a),\beta^+(\bar b)) + 2\epsilon\cdot 3^{-m} &&\mybys{\ref{eqn:approx-coalg-a}}{\ref{eqn:approx-coalg-b}}\\
    &\le Ld^+_\kappa(\alpha^+(\bar a),\beta^+(\bar b)) +
        \epsilon\cdot 3^{-m} + 2\epsilon\cdot 3^{-m} &&\myby{\ref{eqn:approx-ast-eps}}\\
    &= d^+_{\kappa+1}(\bar a,\bar b) + \epsilon\cdot 3^{1-m}. &&\myby{definition of $d^+_{\kappa+1}$}
  \end{align*}
  We can symmetrically derive $d^+_{\kappa+1}(\bar a,\bar b) \le
  d^\epsilon_{\kappa+1}(\bar a,\bar b) + \epsilon\cdot 3^{1-m}$, this
  time using the other inequalities in (\ref{eqn:approx-coalg-a}) and
  (\ref{eqn:approx-coalg-b}), so \eqref{eqn:approx-approximant} holds
  for $\kappa+1$ as claimed. Finally, if $\kappa$ is a limit ordinal,
  then~\eqref{eqn:approx-approximant} also follows inductively, as
  taking suprema is a nonexpansive operation.

  Since the functor $T_\omega$ is finitary, we know from the finitary
  case that $d^\epsilon_\omega = d^\epsilon_{\omega+1}$. Therefore
  we have, for all $\bar a\in A^k, \bar b\in B^k$,
  \begin{equation*}
    |d^+_\omega(\bar a,\bar b) - d^+_{\omega+1}(\bar a,\bar b)|
    \le |d^+_\omega(\bar a,\bar b) - d^\epsilon_\omega(\bar a,\bar b)|
    + |d^\epsilon_{\omega+1}(\bar a,\bar b) - d^+_{\omega+1}(\bar a,\bar b)|
    \le 2\epsilon\cdot 3^{1-k} \le 2\epsilon.  
  \end{equation*}
  Because this holds for all $\epsilon>0$, we have $d^+_\omega = d^+_{\omega+1}$.
  Thus, using Lemma~\ref{lem:approximant}.\ref{item:approx-morph} twice,
  \begin{equation*}
    d_{\omega+1}\circ(\last\times\last) = d^+_{\omega+1}
    = d^+_\omega = d_\omega\circ(\last\times\last).
  \end{equation*}
  As $\last$ is surjective, this implies $d_{\omega+1} = d_\omega$.
\end{proof}

\section{The Kantorovich Lifting}\label{sec:kantorovich}

\noindent As a pseudometric lifting, the Kantorovich lifting is
standard in the probabilistic setting: Given a metric~$d$ on a set~$X$,
the Kantorovich distance $Kd(\mu_1,\mu_2)$ between discrete
distributions~$\mu_1,\mu_2$ on~$X$ is defined by
\begin{equation*}
  Kd(\mu_1,\mu_2) = \sup \{ \expect{\mu_1}(f) - \expect{\mu_2}(f) \mid
    f\colon(X,d)\to([0,1],d_E) \text{ nonexpansive} \}
\end{equation*}
where $\expect{}$ takes expected values and $d_E(x,y) = |x-y|$ is Euclidean
distance. The coalgebraic
generalization of the Kantorovich lifting, both in the pseudometric
setting~\cite{km:metric-bisimulation-games} and in the present setting
of fuzzy relations, is based on fuzzy predicate liftings, a
quantitative analogue of two-valued predicate liftings
(Section~\ref{sec:prelim}) that goes back to work on coalgebraic fuzzy
description logics~\cite{SchroderPattinson11}. Fuzzy predicate
liftings will feature in the generic quantitative modal logics
that we extract from fuzzy lax extensions
(Section~\ref{sec:cml}).

Recall that the \emph{contravariant fuzzy powerset functor}
$\qfun\colon\Set\op\to\Set$ is defined on sets $X$ as
$\qfun X = (X\to[0,1])$ and on functions $f\colon X\to Y$ as
$\qfun f\colon\qfun Y\to\qfun X$, $\qfun f(h) = h\circ f$.

\begin{defn}[Fuzzy predicate liftings]\label{defn:pred-lifting}
  Let $n \in \mathbb{N}$.
  \begin{enumerate}
  \item An \emph{$n$-ary (fuzzy) predicate lifting} is a natural
    transformation
      \begin{equation*}
        \lambda\colon \qfun^n \Rightarrow \qfun\circ T,
      \end{equation*}
      where the exponent $n$ denotes $n$-fold cartesian product.
    \item The \emph{dual} of $\lambda$ is the $n$-ary predicate lifting
      $\bar\lambda$ given by
      \begin{equation*}
        \bar\lambda(f_1,\dots,f_n) = 1-\lambda(1-f_1,\dots,1-f_n).
      \end{equation*}
    \item We call $\lambda$ \emph{monotone} if for all sets $X$ and all functions
      $f_1,\dots,f_n,g_1,\dots,g_n\in\qfun X$ such that
      $f_i\le g_i$ for all $i$,
      \begin{equation*}
        \lambda_X(f_1,\dots,f_n) \le \lambda_X(g_1,\dots,g_n).
      \end{equation*}
    \item We call $\lambda$ \emph{nonexpansive} if for all sets $X$ and all functions
      $f_1,\dots,f_n,g_1,\dots,g_n\in\qfun X$,
      \begin{align*}
        \supnorm{\lambda_X(f_1,\dots,f_n) - \lambda_X(g_1,\dots,g_n)}
        \le \max(\supnorm{f_1-g_1},\dots,\supnorm{f_n-g_n}).
      \end{align*}
  \end{enumerate}
\end{defn}

\begin{rem}\label{rem:eval-function}
  By the Yoneda lemma, unary predicate liftings are equivalent to the
  \emph{evaluation functions} $e\colon\, T[0,1]\to[0,1]$ used in work
  on pseudometric functor
  liftings~\cite{bbkk:coalgebraic-behavioral-metrics,Schroder08} and
  on the generic Wasserstein lifting~\cite{Hofmann07}; more generally,
  an $n$-ary predicate lifting is equivalent to a generalized form of
  evaluation function, of type $T([0,1]^n)\to [0,1]$~\cite{Schroder08}.
  
  More precisely, an evaluation function $e\colon T[0,1]\to[0,1]$
  gives rise to a unary predicate lifting $\lambda_e$ given by
  $\lambda_e(f) = e \circ Tf$. Conversely, the evaluation function
  corresponding to $\lambda\colon \qfun\Rightarrow\qfun\circ T$ is
  $e_\lambda = \lambda_{[0,1]}(\id)$.
  
  In the more general setting with higher arities, an $n$-ary
  evaluation function is a map $e\colon T([0,1]^n)\to[0,1]$,
  giving rise to a predicate lifting
  $ \lambda_e(f_1,\dots,f_n) = e\circ T\langle f_1,\dots f_n\rangle$,
  while for each $n$-ary predicate lifting $\lambda$ the corresponding
  evaluation function is
  $e_\lambda = \lambda_{[0,1]^n} (\pi_1,\dots,\pi_n)$.
\end{rem}

\noindent Before we can show that the Kantorovich lifting is a lax
extension, we first need to generalize it so that it lifts arbitrary
fuzzy relations instead of just pseudometrics. To this end, we
introduce the notion of nonexpansive pairs (a similar idea appears already
in~\cite[Section~5]{Villani08}):

\begin{defn}\label{defn:rel-nonexp}
  Let $R\colon\frel{A}{B}$. A pair $(f,g)$ of functions $f\colon
  A\to[0,1]$ and $g\colon B\to[0,1]$ is
  \emph{$R$-nonexpansive} if $f(a)-g(b)\le R(a,b)$ for all $a\in A,b\in B$.
\end{defn}

\noindent This notion is compatible with our previous use of the term:
When $A=B$ and $d\colon\frel{A}{A}$ is a hemimetric, then
$f\colon (A,d)\to([0,1],\unithemi)$ is nonexpansive in the sense used
so far (cf.~Section~\ref{sec:prelim}) iff the pair $(f,f)$ is
$d$-nonexpansive in the sense defined above. If~$d$ is a pseudometric,
then this is moreover equivalent to~$f$ being nonexpansive as a map
$(A,d)\to([0,1],d_E)$.

Given a function and a fuzzy relation, we can construct a
\emph{nonexpansive companion}:

\begin{defn}\label{defn:companion}
  Let $R\colon\frel{A}{B}$ and $f\colon A\to[0,1]$. Then we define
  $R[f]\colon B\to[0,1]$ by
  \begin{equation*}\textstyle
    R[f](b) = \sup_{a\in A} f(a)\ominus R(a,b)
  \end{equation*}
  (recall from Section~\ref{sec:prelim} that $\ominus$ denotes
  truncated subtraction).
\end{defn}
\noindent We note some basic properties of nonexpansive pairs and nonexpansive
companions. In particular, the nonexpansive companion of some function $f$ is
the least function (in pointwise order) forming a nonexpansive pair with $f$.

\begin{lem}\label{lem:rel-nonexp-props}
  Let $R\colon\frel{A}{B}$. Then the following hold:
  \begin{enumerate}
    \item If $f'\le f$ and $g\le g'$ and $(f,g)$ is $R$-nonexpansive,
      then $(f',g')$ is $R$-nonexpansive.
    \item $(f,g)$ is $R$-nonexpansive if and only if $R[f] \le g$.
  \end{enumerate}
\end{lem}

\begin{defn}\label{defn:kantorovich}
  Let $\Lambda$ be a set of monotone predicate liftings.
  The \emph{Kantorovich lifting} $K_\Lambda$ is defined as follows:
  for $R\colon\frel{A}{B}$, $K_\Lambda R\colon\frel{TA}{TB}$ is given by
  \begin{multline*}
    K_\Lambda R(t_1,t_2) =
      \sup\{\lambda_A(f_1,\dots,f_n)(t_1)\ominus\lambda_B(g_1,\dots,g_n)(t_2) \mid\\
    \hspace{-0.5cm} \lambda\in\Lambda\text{ $n$-ary},
      (f_1,g_1),\dots (f_n,g_n)\text{ $R$-nonexpansive}\}.
  \end{multline*}
\end{defn}

\noindent To show that the Kantorovich lifting is a lax extension, we
need the following fact about nonexpansive pairs that will be used in
the proof of the triangle inequality (L2).

\begin{lem}\label{lem:nonexp-interpol}
  Let $R\colon\frel{A}{B},S\colon\frel{B}{C}$.
  Then for every $(R;S)$-nonexpansive pair $(f,h)$ there
  exists some function $g\colon B\to[0,1]$ such that $(f,g)$ is
  $R$-nonexpansive and $(g,h)$ is $S$-nonexpansive.
\end{lem}
\begin{proof}
  For each $b\in B$ the value $g(b)$ can be chosen arbitrarily
  in the interval
  \begin{align*}
    [\sup_{a\in A} f(a) \ominus R(a,b), \inf_{c\in C} h(c)  \oplus S(b,c)],
  \end{align*}
  so for instance we can use the nonexpansive companion $g := R[f]$
  (Definition~\ref{defn:companion}).
  This interval is non-empty because by assumption
  \begin{equation*}
    f(a)-h(c) \le (R;S)(a,c)
    \le \inf_{b' \in B} R(a,b')+S(b',c)
    \le R(a,b)+S(b,c)
  \end{equation*}
  for all $a\in A,c\in C$, so $f(a)-R(a,b)\le h(c)+S(b,c)$ by
  rearranging. Similar rearranging also shows that choosing $g(b)$
  in this way ensures that $(f,g)$ is $R$-nonexpansive and
  $(g,h)$ is $S$-nonexpansive.
\end{proof}

\noindent We are now ready to prove the central result of the section,
stating that the Kantorovich lifting is always a fuzzy lax extension.
In general, it does not preserve converse, but does if the set of
predicate liftings contains all duals of predicate liftings.

\begin{thm}\label{thm:kantorovich-is-lax}
  Let $\Lambda$ be a set of monotone predicate liftings.
  The Kantorovich lifting $K_\Lambda$ is a lax extension.
  If $\Lambda$ is closed under duals, then $K_\Lambda$ preserves
  converse. If all $\lambda\in\Lambda$ are nonexpansive, then
  $K_\Lambda$ is nonexpansive as well.
\end{thm}
\begin{proof}
  For readability, we pretend that all $\lambda\in\Lambda$
  are unary although the proof works just as well for unrestricted arities,
  whose treatment requires no more than adding indices.
  We show the five properties one by one:

  \begin{itemize}
    \item (L1): Let $R_1\le R_2$. Then every $R_1$-nonexpansive pair
      is also $R_2$-nonexpansive. Thus $K_\Lambda R_1\le K_\Lambda R_2$,
      because the supremum on the left side is taken over a subset of that
      on the right side.
    \item (L2): Let $R\colon\frel{A}{B},S\colon\frel{B}{C}$ and
      $t_1\in TA,t_2\in TB,t_3\in TC$. Let $\lambda\in\Lambda$ and
      let $(f,h)$ be $(R;S)$-nonexpansive. Let $g$ be given by
      Lemma~\ref{lem:nonexp-interpol}. Then it is enough to observe that:
      \begin{align*}
        \lambda_A(f)(t_1)\ominus\lambda_C(h)(t_3)
        &\le (\lambda_A(f)(t_1)\ominus\lambda_B(g)(t_2)) + (\lambda_B(g)(t_2)\ominus\lambda_C(h)(t_3))\\
        &\le K_\Lambda R(t_1,t_2) + K_\Lambda S(t_2,t_3).
      \end{align*}
    \item (L3): Let $h\colon A\to B$ and $t\in TA$. We need to show that
      $K_\Lambda \gr{h}(t,Th(t)) = 0$. Let $\lambda\in\Lambda$ and
      let $(f,g)$ be $\gr{h}$-nonexpansive, implying $f\le g\circ h$. Then
      \begin{align*}
        \lambda_A(f)(t)
        \le \lambda_A(g\circ h)(t)
        = \lambda_B(g)(Th(t)),
      \end{align*}
      by monotonicity and naturality of $\lambda$. The proof for
      $\rev{\gr{h}}$ is analogous, noting that a pair $(f,g)$ is
      $\rev{\gr{h}}$-nonexpansive iff $f\circ h\le g$.
    \item (L4): Let $A$ be a set, $t\in TA$ and $\epsilon>0$. We need to show that
      $K_\Lambda\Delta_{\epsilon,A}(t,t) \le \epsilon$. Let $\lambda\in\Lambda$ and
      let $(f,g)$ be $\Delta_{\epsilon,A}$-nonexpansive, implying $f(a)-g(a) \le \epsilon$
      for all $a\in A$. By monotonicity of $\lambda$, we can restrict our attention to the case
      $g(a) = f(a)\ominus\epsilon$, so that we have $\supnorm{f-g}\le\epsilon$. In this case,
      \begin{align*}
        \lambda_A(f)(t)\ominus\lambda_A(g)(t) \le \supnorm{\lambda_A(f)-\lambda_A(g)}
        \le \supnorm{f-g} \le \epsilon. 
      \end{align*}
    \item (L0): Let $R\colon\frel{A}{B}$ and $t_1\in TA,t_2\in TB$.
      Note that a pair $(g,f)$ is $\rev{R}$-nonexpansive
      iff $(1-f,1-g)$ is $R$-nonexpansive.
      Now, using that $\Lambda$ is closed under duals,
      \begin{align*}
        K_\Lambda (\rev{R})(t_2,t_1)
        &= \sup\{\lambda_B(g)(t_2)\ominus\lambda_A(f)(t_1)
          \mid \lambda\in\Lambda,(g,f)\text{ $\rev{R}$-nonexp.}\} \\
        &= \sup\{\bar\lambda_A(f)(t_1)\ominus\bar\lambda_B(g)(t_2)
          \mid \lambda\in\Lambda,(f,g)\text{ $R$-nonexp.}\}
          = K_\Lambda R(t_1,t_2). \tag*{\qedhere}
      \end{align*}    
  \end{itemize}
\end{proof}

\begin{rem}[Kantorovich for pseudometrics]\label{rem:kant-pseudometric}
  On pseudometrics, the Kantorovich lifting~$K_\Lambda$ as given by
  Definition~\ref{defn:kantorovich} agrees with the usual Kantorovich
  distance $-^{\uparrow T}$ defined for pseudometrics~\cite[Definition
  5.4]{bbkk:coalgebraic-behavioral-metrics}.  If $d\colon\frel{A}{A}$
  is a pseudometric, then
\begin{multline*}
  d^{\uparrow T}(t_1,t_2) = 
  \sup\{|\lambda_A(f_1,\dots,f_n)(t_1)-\lambda_A(f_1,\dots,f_n)(t_2)|\; \mid \\
    \lambda\in\Lambda,
    f_1,\dots,f_n\colon\nonexp{(A,d)}{([0,1],d_E)}\}.
\end{multline*}
\end{rem}

\begin{lem}
  If $\Lambda$ is closed under duals, then $K_\Lambda(d) = d^{\uparrow
  T}$ for every pseudometric $d$.
\end{lem}
\begin{proof}
  First, note that if $(f,g)$ with $f,g\colon A\to[0,1]$ is
  $d$-nonexpansive, then $f(a)-g(a) \le d(a,a) = 0$ for all $a\in A$,
  so $f\le g$.  By monotonicity of the $\lambda\in\Lambda$, the value
  of the supremum in Definition~\ref{defn:kantorovich} thus does not
  change if we restrict the choice of $(f,g)$ to the case $f=g$.
  Finally, in case $f=g$, $d$-nonexpansiveness implies that
  $f(a)-f(b)\le d(a,b)$ and $f(b)-f(a)\le d(b,a) = d(a,b)$ for every
  $a,b\in A$, which means that $f$ is in fact a nonexpansive map
  $f\colon\nonexp{(A,d)}{([0,1],d_E)}$.  Also the supremum does not
  change when taking the absolute value, because $f$ is nonexpansive
  iff $1-f$ is and $\Lambda$ is closed under duals.
\end{proof}

\begin{expl}[Kantorovich liftings]\label{expl:kantorovich}\hfill
  \begin{enumerate}
  \item\label{item:prob-kantorovich} The standard Kantorovich lifting
    $K$ of the discrete distribution functor $\dfun$ is an instance of
    the generic one, for the single predicate lifting
    $\Diamond(f)(\mu) = \expect{\mu}(f)$. We claim that $K$ is
    finitarily separable.
    To see this, let $\mu\in\dfun X$ and $\epsilon>0$. We need to find
    $\mu_\epsilon\in\dfun X$ with finite support such that
    $K\Delta_X(\mu,\mu_\epsilon) \le \epsilon$. Note that a pair
    $(f,g)$ is $\Delta_X$-nonexpansive iff $f\le g$, so by
    monotonicity
      \begin{equation*} \textstyle
        K \Delta_X(\mu,\mu_\epsilon)
        = \sup \{ \sum_{x\in X} f(x) (\mu(x)\ominus\mu_\epsilon(x)) \mid f\colon X\to[0,1] \}
        \le \sum_{x\in X} |\mu(x)-\mu_\epsilon(x)|.
      \end{equation*}
    Because $\mu$ is discrete, there exists a finite set $Y\subseteq X$ with
    $\mu(Y)\ge 1-\frac{\epsilon}{2}$. If $Y = X$, then we can just put
    $\mu_\epsilon = \mu$.
    Otherwise, let $x_0\in X\setminus Y$. Then we define $\mu_\epsilon$ as follows:
    $\mu_\epsilon(x_0) = \mu(X\setminus Y)$,
    $\mu_\epsilon(x) = \mu(x)$ for $x\in Y$,
    and $\mu_\epsilon(x) = 0$ otherwise. In this case,
    \begin{equation*} \textstyle \sum_{x\in X}
      |\mu(x)-\mu_\epsilon(x)| \le 2\mu(X\setminus Y) \le
      \epsilon.
    \end{equation*}
    Following Remark~\ref{rem:diagonals}, we can also see that the
    Kantorovich lifting characterizes behavioural equivalence for
    probabilistic transition systems, i.e.~\emph{probabilistic
    bisimilarity}~\cite{LarsenSkou91}: To see that $K$
    satisfies~\eqref{eq:diag-metric}, by
    Lemma~\ref{lem:functor-lifting} it suffices to show that
    $K\Delta_X(\mu_1,\mu_2)>0$ for any $\mu_1\neq\mu_2\in\dfun X$.
    W.l.o.g. assume $\mu_1(x)>\mu_2(x)$ for some $x\in X$ and let $f\in\qfun X$ be such
    that $f(x) = 1 $ and $f(x') = 0$ otherwise. Then, as $(f,f)$ is
    $\Delta_X$-nonexpansive, we have $K\Delta_X(\mu_1,\mu_2) \ge
    f(x)(\mu_1(x)-\mu_2(x)) > 0$.

  \item The asymmetric Hausdorff lifting $H^\leftarrow$
    (Example~\ref{expl:hausdorff}) is equal to the Kantorovich lifting
    for the single predicate lifting $\Diamond_X(f)(A) = \sup f[A]$.
    Let $R\colon\frel{A}{B}$ and let $U\subseteq A$, $V\subseteq B$.
    We show $H^\leftarrow R(U,V) = K_{\{\Diamond\}}R(U,V)$.
    \begin{itemize}
    \item `$\le$': Let $(f,g)$ be an $R$-nonexpansive pair. Then
      \begin{equation*}
        \adjustlimits \sup_{a\in U} f(a) \ominus \sup_{b\in V} g(b)
          \le \adjustlimits \sup_{a\in U} \inf_{b\in V} f(a) \ominus g(b)
          \le \adjustlimits \sup_{a\in U} \inf_{b\in V} R(a,b) = H^\leftarrow R(a,b).
        \end{equation*}
      \item `$\ge$': Let $a\in U$ and let $f\in\qfun A$ be the indicator
        function of $\{a\}$, that is $f(a') = 1$ if $a'=a$ and $f(a') = 0$
        otherwise. Put $g = R[f]$, so that~$g(b) = 1\ominus R(a,b)$ for each $b\in B$.
        Then, as $(f,g)$ is $R$-nonexpansive (Lemma~\ref{lem:rel-nonexp-props}),
        \begin{equation*}
          K_{\{\Diamond\}}R(U,V) \ge
          \adjustlimits \sup_{a\in U} f(a) \ominus \sup_{b\in V} g(b)
          = 1 \ominus \sup_{b\in V} (1 \ominus R(a,b)) = \inf_{b\in V} R(a,b).
        \end{equation*}
      \end{itemize}
      Dually, the other asymmetric form $H^\rightarrow$ of the
      Hausdorff lifting is thus the Kantorovich lifting for the
      single predicate lifting $\Box_X(f)(A)=\inf f[A]$. It follows
      immediately that the symmetric Hausdorff lifting~$H$ is the
      Kantorovich lifting $K_\Lambda$ for $\Lambda=\{\Box,\Diamond\}$.
  \item The \emph{fuzzy neighbourhood functor} is the (covariant)
    functor $\nfun = \qfun\circ\qfun$; the elements of $\nfun X$ are
    called \emph{fuzzy neighbourhood systems}, and their coalgebras
    \emph{fuzzy neighbourhood
      frames}~\cite{RodriguezGodo13,CintulaEA16}.  The \emph{monotone
      (nonexpansive) fuzzy neighbourhood functor}~$\mfun$ is the
    subfunctor~$\mfun$ of $\nfun$ given by $\mfun X$ consisting of the
    fuzzy neighbourhood systems that are monotone and nonexpansive as
    maps $A\colon\qfun X\to[0,1]$. We put
    \begin{equation*} \textstyle
      LR(A,B) = \sup_{f\in\qfun X} A(f)\ominus B(R[f])
    \end{equation*}
    \noindent for $R\colon\frel{X}{Y}$, $A\in\mfun X$,
    $B\in\mfun Y$ (recall Definition~\ref{defn:companion}). Then
    $L$ is a nonexpansive lax extension of $\mfun$; specifically,
    $L = K_{\{\lambda\}}$ where $\lambda$ is the predicate lifting
    given by
    \begin{math}
      \lambda_X(f)(A) = A(f).
    \end{math}
  \end{enumerate}
\end{expl}

\section{The Wasserstein Lifting}\label{sec:wasserstein}

The other generic construction for lax extensions arises in a similar
way, by generalizing the generic Wasserstein lifting for
pseudometrics~\cite{bbkk:coalgebraic-behavioral-metrics} to lift
arbitrary fuzzy relations instead of just pseudometrics; our
construction slightly generalizes one given by
Hofmann~\cite{Hofmann07}. Like the Kantorovich lifting, the
Wasserstein lifting is based on a choice of predicate liftings.
Compared to the case of the Kantorovich lifting, where we needed to
work with nonexpansive pairs, the generalization from lifting
pseudometrics to lifting relations is much more direct for the
Wasserstein lifting. In the same way as for the original
construction of pseudometric Wasserstein liftings, additional
constraints, both on the functor and the set of predicate liftings
involved, are needed for the Wasserstein lifting to be a lax
extension. Indeed, the Wasserstein lifting may be seen as a
quantitative analogue of the two-valued Barr extension
(Section~\ref{sec:prelim}), and like the latter works only for
functors that preserve weak pullbacks. In particular, Wasserstein
liftings are based on the central notion of coupling:

\begin{defn}
  Let $t_1\in TA,t_2\in TB$ for sets $A,B$. The set of
  \emph{couplings} of $t_1$ and $t_2$ is
  $\cpl{t_1}{t_2}=\{t\in T(A\times B)\mid T\pi_1(t) = t_1, T\pi_2(t) =
  t_2\}$.
\end{defn}
\noindent The Wasserstein lifting uses predicate liftings in a quite
different manner from the Kantorovich lifting, and in particular
appears to make sense only for unary predicate liftings, so unlike
elsewhere in the paper, the restriction to unary liftings in
the next definition is not just for readability.

\begin{defn}[Wasserstein lifting]
  Let $\Lambda$ be a set of unary predicate liftings.  The
  \emph{generic Wasserstein lifting} is the relation lifting~$W_\Lambda$
  of~$T$ defined for $R\colon \frel{A}{B}$ by
  \begin{align*} \textstyle
    W_\Lambda R(t_1,t_2) =
      \sup_{\lambda\in\Lambda}\,
      \inf\{\lambda_{A\times B}(R)(t) \mid t \in \cpl{t_1}{t_2} \}.
  \end{align*}
\end{defn}
\noindent This construction is similar
to~\cite[Definition~3.4]{Hofmann07} except that we admit more than one
modality. On pseudometrics, the Wasserstein lifting coincides with the
pseudometric lifting $-^{\downarrow T}$ as defined in~\cite[Definition
5.12]{bbkk:coalgebraic-behavioral-metrics} (again up to the fact that
we admit more than one modality).  We will see that the following
conditions ensure that the Wasserstein lifting is a fuzzy lax
extension:
\begin{defn}
  Let $\lambda$ be a unary predicate lifting.
  \begin{enumerate}
    \item $\lambda$ is \emph{subadditive} if for all sets $X$ and all
      $f,g\in\qfun X$,
      $\lambda_X(f\oplus g) \le \lambda_X(f) \oplus \lambda_X(g)$.
    \item $\lambda$ \emph{preserves the zero function} if for all sets $X$,
      $\lambda_X(\zerofun{X}) = \zerofun{TX}$,
      where $\zerofun{X}\colon x\mapsto 0$.
    \item $\lambda$ is \emph{standard} if it is monotone,
      subadditive, and preserves the zero function.
  \end{enumerate}
\end{defn}

\noindent Baldan et al. give conditions under which the Wasserstein
lifting arising from some set of evaluation functions
(Remark~\ref{rem:eval-function}) preserves pseudometrics. For this
purpose they consider the notion of a \emph{well-behaved evaluation
function}~\cite[Definition 5.14]{bbkk:coalgebraic-behavioral-metrics}.

\begin{defn}\label{defn:well-behaved}
  An evaluation function $e\colon T[0,1]\to[0,1]$ is  \emph{well-behaved}
  if it satisfies the following conditions.
  \begin{enumerate}
    \item The predicate lifting $\lambda_e$ is monotone.
    \item For all $t\in T([0,1]^2)$, we have
      $d_E(e(t_1),e(t_2)) \le \lambda_e(d_E)(t)$, where
      $t_j = T\pi_j(t)$ for $j=1,2$. 
    \item $e^{-1}[\{0\}] = Ti[T\{0\}]$,
      where $i\colon\{0\}\to[0,1]$ is the inclusion map.
  \end{enumerate}
\end{defn}

\noindent This amounts to a slightly stronger condition than
standardness of the corresponding predicate lifting:

\begin{lem}\label{lem:wb-equals-standard}
  An evaluation function $e\colon T[0,1]\to[0,1]$ is well-behaved
  iff the predicate lifting~$\lambda_e$ is standard and
  $e^{-1}[\{0\}] \subseteq Ti[T\{0\}]$.
\end{lem}
\begin{proof}
  First, note that monotonicity of $\lambda_e$ features in both notions
  and $\lambda_e$ preserves zero iff $e^{-1}[\{0\}] \supseteq Ti[T\{0\}]$.
  It remains to relate Item 2 of Definition~\ref{defn:well-behaved} with subadditivity of $\lambda_e$.
  Reformulating in terms of $\lambda_e$ gives
  \begin{equation}
    |\lambda_e(\pi_1)(t) - \lambda_e(\pi_2)(t)| \le \lambda_e(d_E)(t)
    \quad\text{for } t\in T([0,1]^2). \label{eq:wb2}
  \end{equation}
  We show that \eqref{eq:wb2} is equivalent to subadditivity of $\lambda_e$,
  given that $\lambda_e$ is monotone:
  \begin{itemize}
    \item `$\Rightarrow$':
      Let $f,g\in\qfun X, t\in TX$.
      Put $t' := T\langle f\oplus g,f\rangle(t)\in T([0,1]^2)$.
      Then, by naturality, we have $\lambda_e(\pi_1)(t') = \lambda_e(f\oplus g)(t)$
      and $\lambda_e(\pi_2)(t') = \lambda_e(f)(t)$ and
      \begin{equation*}
        \lambda_e(d_E)(t') = \lambda_e(d_E\circ\langle f\oplus g,f\rangle)(t)
        \le \lambda_e(g)(t),
      \end{equation*}
      where the last step is by monotonicity of $\lambda_e$.
      Therefore,
      $\lambda(f\oplus g)(t) - \lambda(f)(t) \le \lambda(g)(t)$ by
      \eqref{eq:wb2}.
    \item `$\Leftarrow$': Put $f = d_E, g = \pi_1\colon [0,1]^2\to[0,1]$.
      Then it is easily checked that $f\oplus g \ge \pi_2$ and therefore
      \begin{equation*}
        \lambda_e(\pi_2) \le \lambda_e(f\oplus g) \le \lambda_e(f) + \lambda_e(g)
        = \lambda_e(d_E) + \lambda_e(\pi_1)
      \end{equation*}
      by monotonicity and subadditivity of $\lambda_e$, so 
      $\lambda_e(\pi_1) - \lambda_e(\pi_2) \le \lambda_e(d_E)$.
      Similarly, we can show that 
      $\lambda_e(\pi_2) - \lambda_e(\pi_1) \le \lambda_e(d_E)$
      by swapping the roles of $\pi_1$ and $\pi_2$. \hfill\qedhere
  \end{itemize}
\end{proof}

\noindent Similar conditions also feature in Hofmann's
\emph{topological theories}~\cite[Definition 3.1]{Hofmann07}, which
consist of a monad acting on a quantale via an evaluation function and
on which his generic Wasserstein extension is based.
Explicitly, a topological theory
is defined as a triple consisting of a monad $T$, a quantale $V$, and
a map $\xi\colon TV\to V$ satisfying a number of axioms. We only
consider the case of the quantale $[0,1]^\text{op}$, with the order
given by $\ge$ and the monoid structure by $\oplus$. The first two
axioms state that $\xi$ is a $T$-algebra and can be ignored for our
purposes. The remaining axioms instantiate as follows, where as usual
$\lambda_\xi(f) = \xi\circ Tf$ is the predicate lifting associated
with $\xi$:
\begin{align*}
  (Q_\otimes) &\qquad
    \oplus \circ \langle\lambda_\xi(\pi_1),\lambda_\xi(\pi_2)\rangle \ge \lambda_\xi(\oplus) \\
  (Q_k) &\qquad
    0 \ge \lambda_\xi(\zerofun{1})(t) \text{ for every $t\in T1$, where $1$ is a singleton set} \\
  (Q'_\vee) &\qquad
    \lambda_\xi \text{ is a monotone natural transformation}
\end{align*}

\noindent Using a similar idea as in
Lemma~\ref{lem:wb-equals-standard}, we see that $(Q_\otimes)$ is
equivalent to subadditivity of~$\lambda_\xi$ and~$(Q_k)$ is equivalent
to preservation of the zero function. Finally note that~\cite[Theorem
3.5 (d)]{Hofmann07} (which states that the Wasserstein lifting
satisfies (L2)) requires that the functor satisfies the
\emph{Beck-Chevalley condition}, i.e.  preserves weak pullbacks.

If $T$ preserves weak pullbacks, the following so-called gluing lemma
holds~\cite[Lemma 5.18]{bbkk:coalgebraic-behavioral-metrics}:

\begin{lem}[Gluing]\label{lem:glue}
  Let $A$, $B$ and $C$ be sets, and let
  $t_1\in TA, t_2\in TB, t_3\in TC$.  Let $t_{12}\in\cpl{t_1}{t_2}$
  and $t_{23}\in\cpl{t_2}{t_3}$.  Then there exists
  $t_{123}\in\cpltri{t_1}{t_2}{t_3}$ such that
  \begin{equation*}
    T\langle\pi_1,\pi_2\rangle(t_{123}) = t_{12}
    \quad\text{and}\quad
    T\langle\pi_2,\pi_3\rangle(t_{123}) = t_{23},
  \end{equation*}
  where the $\pi_j$ are the projections of the product $A\times B\times C$.
  Moreover, $t_{13} := T\langle\pi_1,\pi_3\rangle(t_{123}) \in\cpl{t_1}{t_3}$.
\end{lem}

\noindent Using Lemma~\ref{lem:glue}, we can now show that the Wasserstein lifting
is a fuzzy lax extension. Unlike the Kantorovich lifting, the
Wasserstein lifting always preserves converse, without any further restrictions
on the set of predicate liftings.

\begin{thm}\label{thm:wasserstein-is-lax}
  If $T$ preserves weak pullbacks and $\Lambda$ is a set of standard
  predicate liftings, then
  the Wasserstein lifting $W_\Lambda$ is a converse-preserving lax extension.
  If additionally all $\lambda\in\Lambda$ are nonexpansive, then $W_\Lambda$
  is nonexpansive as well.
\end{thm}

\begin{proof}
  We show the five properties one by one:
  \begin{itemize}
    \item (L0): Let $\swap = \langle \pi_2,\pi_1 \rangle \colon
      A\times B\to B\times A$. Then $T\swap$ is an isomorphism
      between $\cpl{t_1}{t_2}$ and $\cpl{t_2}{t_1}$ and it suffices to observe
      that for every $\lambda\in\Lambda$ and $t\in T(A\times B)$,
      \begin{math}
        \lambda_{B\times A}(\rev{R})(T\swap(t)) =
        \lambda_{A\times B}(R)(t)
      \end{math}
      by naturality of $\lambda$.
    \item (L1): Immediate from the definition of $W_\Lambda$ and monotonicity of the
      predicate liftings.
    \item (L2): Let $R\colon\frel{A}{B}, S\colon\frel{B}{C}$ and let
      $t_1\in TA, t_2\in TB, t_3\in TC$. We need to show that
      $W_\Lambda (R;S)(t_1,t_3) \le W_\Lambda R(t_1,t_2) + W_\Lambda S(t_2,t_3)$.
      Let $\lambda\in\Lambda$, $t_{12}\in\cpl{t_1}{t_2}$ and $t_{23}\in\cpl{t_2}{t_3}$,
      and let $t_{123}$ and $t_{13}$ be as in Lemma~\ref{lem:glue}. We need to show
      \begin{align}
        \lambda_{A\times C}(R;S)(t_{13}) \le
        \lambda_{A\times B}(R)(t_{12}) + \lambda_{B\times C}(S)(t_{23}). \label{eq:triangle-wasserstein}
      \end{align}
      We define three functions $f_{12},f_{13},f_{23}\colon A\times B\times C\to[0,1]$ by
      $f_{12}(a,b,c) = R(a,b)$, $f_{23}(a,b,c) = S(b,c)$, and $f_{13}(a,b,c) = (R;S)(a,c)$.
      Then, as $f_{13} \le f_{12} \oplus f_{23}$, we obtain
      \begin{equation*}
        \lambda_{A\times B\times C}(f_{13})(t_{123}) \le
          \lambda_{A\times B\times C}(f_{12})(t_{123}) +
            \lambda_{A\times B\times C}(f_{23})(t_{123})
      \end{equation*}
        by monotonicity and subadditivity of $\lambda$, which is equivalent to
        (\ref{eq:triangle-wasserstein}) by naturality of $\lambda$.
    \item (L3):
      Let $f\colon A\to B$, $t_1\in TA$ and $\lambda\in\Lambda$.
      We need to find $t\in\cpl{t_1}{Tf(t_1)}$ such that
      $\lambda_{A\times B}(\gr{f})(t) = 0$.
      Indeed, take $t = T\langle\id_A,f\rangle(t_1)$.
      Then $T\pi_1(t) = T\id_A(t_1) = t_1$ and $T\pi_2(t) = Tf(t_1)$, and,
      as $\lambda$ is natural and preserves zero,
      \begin{equation*}
        \lambda_{A\times B}(\gr{f})(t) = \lambda_A(\gr{f}\circ\langle\id_A,f\rangle)(t_1)
        = \lambda_A(\zerofun{A})(t_1) = 0.
        \end{equation*}
      This proves $W_\Lambda(\gr{f}) \le \gr{Tf}$. The second clause $W_\Lambda(\rev{\gr{f}}) \le \rev{\gr{Tf}}$ now follows using (L0).
      \item (L4): Let $A$ be a set, $\epsilon>0$, $t_1\in TA$ and
        $\lambda\in\Lambda$.  It is enough to find
        $t\in\cpl{t_1}{t_1}$ such that
        $\lambda_{A\times A}(\Delta_{\epsilon,A})(t) \le\epsilon$.
        Indeed, take $t = T\langle\id_A,\id_A\rangle(t_1)$. Then
        $T\pi_1(t) = T\pi_2(t) = t_1$, and with $\epsilon_A\colon A\to[0,1]$
        being the constant map $a\mapsto\epsilon$ we derive
      \begin{align*}
        \lambda_{A\times A}(\Delta_{\epsilon,A})(t)
        = \lambda_A(\constfun{\epsilon}{A})(t_1)
        \le \supnorm{\lambda_A(\constfun{\epsilon}{A}) - \lambda_A(\zerofun{A})}
        \le \epsilon,
      \end{align*}
      using that $\lambda$ is natural, nonexpansive and preserves zero.\hfill\qedhere
  \end{itemize}
\end{proof}

\begin{expl}[Wasserstein liftings]\label{expl:wasserstein}\hfill
  \begin{enumerate}
  \item \label{item-kr-duality} Similar to the case of the standard
    Kantorovich lifting $K$
    (Example~\ref{expl:kantorovich}.\ref{item:prob-kantorovich}), the
    standard Wasserstein lifting $W$ of the discrete distribution
    functor $\dfun$ arises as an instance of the generic
    Wasserstein lifting, for the same predicate lifting
    $\Diamond(f)(\mu) = \expect{\mu}(f)$. In fact, it is well
    known~\cite[Theorem 5.10]{Villani08} that $K = W$, a fact known
    as \emph{Kantorovich-Rubinstein duality}.
  \item \label{item:hausdorff-is-lax} The Hausdorff lifting~$H$
    (Example~\ref{expl:hausdorff}) is the Wasserstein lifting
    $W_{\{\lambda\}}$ for $\Pow$, where $\lambda_X(f)(A) = \sup f[A]$
    for $A\subseteq X$. To see this, let $R\colon\frel{A}{B}$, and
    let $U\subseteq A$ and $V\subseteq B$. Then we show that $HR(U,V)
    = W_{\{\lambda\}}R(U,V)$ by proving the two inequalities separately:
    \begin{itemize}
      \item `$\le$': Let $Z\in\cpl{U}{V}$. Then for every $a\in U$ there
        exists $b\in V$ such that $(a,b)\in Z$, so $\inf_{b\in V} R(a,b)
        \le \sup R[Z]$. Thus, we have $\sup_{a\in U}\inf_{b\in V} R(a,b)
        \le \sup R[Z]$, and, by a symmetrical argument, $\sup_{b\in
        V}\inf_{a\in U} R(a,b) \le \sup R[Z]$.
      \item `$\ge$': Let $\epsilon>0$. It suffices to find a coupling
        $Z\in\cpl{U}{V}$ such that $\sup R[Z] \le HR(U,V) + \epsilon$. So
        let $\epsilon>0$. We construct functions $f\colon U\to V$ and
        $g\colon V\to U$ as follows: For each $a\in U$ choose $f(a)\in V$
        such that $R(a,f(a)) \le \inf_{b\in V} R(a,b) + \epsilon$.
        Similarly, for each $b\in V$ choose $g(b)\in U$ such that
        $R(g(b),b) \le \inf_{a\in U} R(a,b) + \epsilon$. Now put $Z =
        \{(a,f(a)) \mid a\in U\} \cup \{(g(b),b) \mid b\in V\}$. Clearly,
        $Z\in\cpl{U}{V}$ and by construction,
        \begin{equation*} \textstyle
          \sup R[Z] = \max(\sup_{a\in U} R(a,f(a)), \sup_{b\in V} R(g(b),b))
          \le HR(U,V) + \epsilon.
        \end{equation*}
    \end{itemize}

  \item\label{item:convex-powerset} The convex powerset functor
    $\Conv$, whose coalgebras combine probabilistic branching and
    nondeterminism~\cite{BonchiEA17}, maps a set $X$ to the set of
    nonempty convex subsets of~$\dfun X$. The Wasserstein
    lifting~$W_{\{\lambda\}}$, where $\lambda_X(f)(A) = \sup_{\mu\in
    A} \expect{\mu}(f)$ for $A\in\Conv X$, is a nonexpansive lax
    extension of~$\Conv$. Of course,~$\lambda$ is just the composite
    of the predicate liftings respectively defining the standard
    Kantorovich/Wasserstein and Hausdorff liftings.
    $W_{\{\lambda\}}$ indeed coincides with the composite $HW$ of these
    liftings (for which a quantitative equational axiomatization has
    recently been given by Mio and Vignudelli~\cite{MioVignudelli20}):
    
    Let $R\colon\frel{A}{B}$, and let $U\in\Conv A$ and $V\in\Conv
    B$. We show $W_{\{\lambda\}}(R)(U,V) = HW(R)(U,V)$. There are two
    inequalities:
    \begin{itemize}
      \item `$\ge$': Let $Z\in\cplind{\Conv}{U}{V}$. We put $Y =
        \Pow\langle\dfun\pi_1,\dfun\pi_2\rangle(Z)$. Then $\Pow\pi_1(Y) =
        \Pow\dfun\pi_1(Z) = \Conv\pi_1(Z) = U$ and similarly $\Pow\pi_2(Y)
        = V$, so that $Y\in\cplind{\Pow}{U}{V}$. Now, note that for every
        $\mu\in\dfun(A\times B)$ we have that
        $\expect{\mu}(R) \ge WR(\dfun\pi_1(\mu),\dfun\pi_2(\mu))$ and
        therefore
        \begin{equation*}
          \sup_{\mu\in Z}\expect{\mu}(R)
          \ge \sup_{(\mu_1,\mu_2)\in Y} WR(\mu_1,\mu_2)
          \ge HW(R)(U,V).
        \end{equation*}
      \item `$\le$': Let $Y\in\cplind{\Pow}{U}{V}$ and
        $\epsilon>0$. It suffices to find
        $Z\in\cplind{\Conv}{\mu_1}{\mu_2}$ such that
        \begin{equation*}
          \sup_{\mu\in Z}\expect{\mu}(R)
          \le \sup_{(\mu_1,\mu_2)\in Y} WR(\mu_1,\mu_2) + \epsilon.
        \end{equation*}
        For every $(\mu_1,\mu_2)\in\dfun A\times\dfun B$ there exists some
        $\mu\in\cplind{\dfun}{U}{V}$ such that $\expect{\mu}(R)\le
        WR(\mu_1,\mu_2)+\epsilon$. Let $Z'$ be a set consisting of one
        such $\mu$ for every pair $(\mu_1,\mu_2)\in Y$ and put
        $Z=\convx(Z')$, where $\convx$ is convex hull. Then we have
        \begin{equation*}
          \Conv\pi_1(Z)
          = \Pow\dfun\pi_1(\convx(Z'))
          = \convx (\Pow\dfun\pi_1(Z')) = \convx(U) = U.
        \end{equation*}
        Here we made use of the fact that $\dfun\pi_1$ is linear when
        considered as a map $\mathbb{R}^{A\times B} \to \mathbb{R}^A$, and
        linear maps preserve convex sets. We similarly have
        $\Conv\pi_2(Z)=V$, so that $Z\in\cplind{\Conv}{U}{V}$. Finally, we
        note that taking expected values is a linear operation, so if $\mu
        = \sum_{i=1}^n p_i\mu_i$ is a convex combination of probability
        measures, then $\expect{\mu} = \sum_{i=1}^n p_i\expect{\mu_i} \le
        \max_{i=1}^n \expect{\mu_i}$. Therefore we have, as required,
        \begin{equation*}
          \sup_{\mu\in Z}\expect{\mu}(R)
          = \sup_{\mu\in Z'}\expect{\mu}(R)
          \le \sup_{(\mu_1,\mu_2)\in Y} WR(\mu_1,\mu_2) + \epsilon.\qedhere
        \end{equation*}
      \end{itemize}
  \end{enumerate}
\end{expl}

\section{Lax Extensions as Kantorovich
  Liftings}\label{sec:lax-kantorovich}

We proceed to establish the central result that every fuzzy lax
extension is a Kantorovich lifting for some suitable set $\Lambda$ of
predicate liftings, and moreover we characterize the Kantorovich
liftings induced by nonexpansive predicate liftings as precisely the
nonexpansive lax extensions. For a given fuzzy lax extension~$L$, the
equality $K_\Lambda R = LR$ splits into two inequalities, one of which
is characterized straightforwardly:

\begin{defn}
  An $n$-ary predicate lifting~$\lambda$ \emph{preserves
    nonexpansiveness} if for all fuzzy relations~$R$ and all
  $R$-nonexpansive pairs $(f_1,g_1),\dots,(f_n,g_n)$, the pair
  \begin{equation*}
    (\lambda_A(f_1,\dots,f_n), \lambda_B(g_1,\dots,g_n))
  \end{equation*}
  is $LR$-nonexpansive. A set~$\Lambda$ of predicate liftings
  \emph{preserves nonexpansiveness} if all $\lambda\in\Lambda$ preserve
  nonexpansiveness.
\end{defn}
\begin{lem}\label{lem:adequacy}
  We have $K_\Lambda R \le LR$ for all fuzzy relations $R$ if and only if $\Lambda$ preserves
  nonexpansiveness.
\end{lem}
\begin{defn}[Separation]\label{def:separation}
  A set $\Lambda$ of predicate liftings is \emph{separating} for $L$
  if $K_\Lambda R \ge LR$ for all fuzzy relations $R$.
\end{defn}
\noindent To motivate Definition~\ref{def:separation}, recall from
Section~\ref{sec:prelim} that in the two-valued setting a set
$\Lambda$ of predicate liftings (for simplicity, assumed to be unary)
is separating if
\begin{equation*}
  t_1\neq t_2 \implies
  \exists\lambda\in\Lambda, A'\subseteq A\text{ such that }
  t_1 \in\lambda_A(A') \not\leftrightarrow
  t_2 \in\lambda_A(A')
\end{equation*}
for $t_1,t_2\in TA$. Analogously, unfolding definitions in the
inequality $K_\Lambda R \ge LR$ (and again assuming unary liftings),
we arrive at the condition that for all $t_1\in TA, t_2\in TB$,
$\epsilon > 0$,
\begin{equation*}
  LR(t_1,t_2) > \epsilon\implies\exists \lambda\in\Lambda,
  (f,g)\text{ $R$-nonexpansive such that } \lambda_A(f)(t_1) - \lambda_B(g)(t_2) > \epsilon.
\end{equation*}
\noindent We are now ready to state our main result, which says that
all lax extensions are Kantorovich:

\begin{thm}\label{thm:kantorovich-lax}
  If $L$ is a finitarily separable lax extension of $T$, then there
  exists a set~$\Lambda$ of monotone predicate liftings that preserves
  nonexpansiveness and is separating for $L$, i.e.\ $L=K_\Lambda$.
  Moreover,~$L$ is nonexpansive iff $\Lambda$ can be chosen in such a
  way that all $\lambda\in\Lambda$ are nonexpansive.
\end{thm}
\noindent This result can be seen as a fuzzy version of the statements
that every finitary functor has a separating set of two-valued
modalities (and hence an expressive two-valued coalgebraic modal
logic)~\cite[Corollary~45]{Schroder08}, and that more specifically,
every finitary functor equipped with a diagonal-preserving lax
extension has a separating set of two-valued \emph{monotone} predicate
liftings~\cite[Theorem~14]{MartiVenema15}. We will detail in
Section~\ref{sec:cml} how Theorem~\ref{thm:kantorovich-lax} implies
the existence of characteristic modal logics. The proof of
Theorem~\ref{thm:kantorovich-lax} uses a quantitative version of the
so-called Moss modalities~\cite{KurzLeal09,MartiVenema15}. The
construction of these modalities relies on the fact that~$T_\omega$
can be presented by algebraic operations of finite arity:
\begin{defn}
  A \emph{finitary presentation} of $T_\omega$ consists of a
  \emph{signature} $\Sigma$ of operations with given finite arities,
  and for each $\sigma\in\Sigma$ of arity~$n$ a natural transformation
  $\sigma\colon (-)^n \Rightarrow T_\omega$ such that every element of~$T_\omega X$
  has the form $\sigma_X(x_1,\dots,x_n)$ for some $\sigma\in\Sigma$.
\end{defn}
\noindent For the remainder of this section, we fix a finitary
presentation of~$T_\omega$ with signature~$\Sigma$ (such a
presentation always exists~\cite[Example~21]{MartiVenema15}) and assume a finitarily separable fuzzy
lax extension~$L$ of~$T$. To derive predicate liftings from the
operations in $\Sigma$, we make use of the fuzzy elementhood relation
$\elem_X$ (indexed over arbitrary sets $X$), where
$\elem_X\colon\frel{X}{\qfun X}$ is given by $\elem_X(x,f) = f(x)$.

\begin{defn}\label{defn:moss-lifting}
  Let $\sigma\in\Sigma$ be $n$-ary. The \emph{Moss lifting}
  $\mu^\sigma\colon\qfun^n\Rightarrow \qfun\circ T$ is defined by
  \begin{equation*}
    \mu^\sigma_X(f_1,\dots,f_n)(t) = L\elem_X(t,\sigma_{\qfun X}(f_1,\dots,f_n)).
  \end{equation*}
\end{defn}

\noindent It follows from naturality of $\sigma$ and $L$
(Lemma~\ref{lem:transform}) that $\mu^\sigma$ is indeed
natural and therefore a predicate lifting, as shown next.
Indeed, for any $g\colon A\to B$, $f_1,\dots,f_n\in\qfun B$ and
$t\in TB$ we note that $\elem_A\circ(\id\times\qfun g)=\elem_B\circ(g\times\id)$
by definition of $\elem_A$ and $\elem_B$ and thus
\begin{align*}
  &\mu^\sigma_A(f_1\circ g,\dots,f_n\circ g)(t)\\
  &= L\elem_A(t,\sigma_{\qfun A}(f_1\circ g,\dots,f_n\circ g)) &&\myby{definition of $\mu^\sigma$}\\
  &= L\elem_A(t,T\qfun g(\sigma_{\qfun B}(f_1,\dots,f_n))) &&\myby{$\sigma$ natural}\\
  &= L(\elem_A\circ(\id\times\qfun g))(t,\sigma_{\qfun B}(f_1,\dots,f_n)) &&\myby{Lemma~\ref{lem:transform}}\\
  &= L(\elem_B\circ(g\times\id))(t,\sigma_{\qfun B}(f_1,\dots,f_n))\\
  &= L\elem_B(Tg(t),\sigma_{\qfun B}(f_1,\dots,f_n)) &&\myby{Lemma~\ref{lem:transform}}\\
  &= \mu^\sigma_B(f_1,\dots,f_n)(Tg(t)). &&\myby{definition of $\mu^\sigma$}
\end{align*}
\noindent
We are now in a position to present the proof of
Theorem~\ref{thm:kantorovich-lax}: We take $\Lambda$ to be the set of
Moss liftings and show the required properties of $\Lambda$ one by one.

\begin{conv}
  Throughout this proof, all
  statements and proofs will be written for the case where all
  $\sigma\in\Sigma$ (and therefore the induced Moss liftings) are
  unary.  This is purely in the interest of readability; the general
  case requires only more indexing.
\end{conv}

\paragraph{Monotonicity} The proof is based on the following auxiliary
fact about pairs of elements that are mapped to $0$.

\begin{lem}\label{lem:operator-preserves-dist-0}
  Let $\sigma\in\Sigma$ and $R\colon \frel{A}{B}$. Then
  for all $a\in A,b\in B$ with $R(a,b) = 0$ we have $LR(\sigma_A(a),\sigma_B(b)) = 0$.
\end{lem}
\begin{proof}
  Put $R_0 = \{(a,b) \in A\times B \mid R(a,b) = 0\}$ and
  consider the projection maps $\pi_1\colon R_0\to A$ and
  $\pi_2\colon R_0\to B$. Then it is easy to see that
  $R \le \rev{\gr{\pi_1}};\gr{\pi_2}$ (noting again that we read~$0$
  as `related' and~$1$ as `unrelated'; in particular recall
  Convention~\ref{conv:comp} and Definition~\ref{def:gr}). Using the
  axioms of lax extensions, we obtain 
  \begin{equation}\label{eq:lr-estimate}
    LR \le L(\rev{\gr{\pi_1}};\gr{\pi_2}) \le L \rev{\gr{\pi_1}}; L \gr{\pi_2}
    \le \rev{\gr{T\pi_1}};\gr{T\pi_2}.
  \end{equation}
  For $(a,b)\in R_0$, put $t = \sigma_{A\times B}((a,b))$, so that $T\pi_1(t) =
  \sigma_A(a)$ and $T \pi_2(t) = \sigma_B(b)$ by naturality of $\sigma$.
  This means that
  $(\rev{\gr{T\pi_1}};\gr{T\pi_2})(\sigma(a),\sigma(b)) = 0$,
  so that by~\eqref{eq:lr-estimate} we have $LR(\sigma_A(a),\sigma_B(b)) = 0$.
\end{proof}

\begin{lem}\label{lem:moss-monotone}
  Let $\sigma\in\Sigma$. Then the Moss lifting $\mu^\sigma$ is monotone.
\end{lem}
\begin{proof}
  We make use of the fuzzy relation $R\colon\frel{\qfun X}{\qfun X}$
  given by $R(g,f) = \sup_{x\in X} f(x)\ominus g(x)$,
  which we claim to satisfy the following two useful properties:
  \begin{gather}
    R(g,f) = 0 \iff f \le g \label{eq:ominus-le} \\
    \elem_X \;\;\le\;\; \elem_X;R \label{eq:elem-ominus}
  \end{gather}
  The first property is clear; the second property amounts to showing
  that $f(x) \le g(x) \oplus R(g,f)$ for all $x\in X$ and all
  $f,g\in\qfun X$ and is easily shown by case analysis on the
  definition of $\oplus$.

  Let $f,g\in\qfun X$ with $f\le g$ and let $t\in TX$. First, we note
  that by \eqref{eq:ominus-le} we have $R(g,f) = 0$ and thus
  $LR(\sigma_{\qfun X}(g),\sigma_{\qfun X}(f)) = 0$ by
  Lemma~\ref{lem:operator-preserves-dist-0}. Second, by
  \eqref{eq:elem-ominus} and the axioms of lax extensions we have
  $L\elem_X \le L(\elem_X;R) \le L\elem_X;LR$. Therefore:
  \begin{align*}
    \mu^\sigma_X(f)(t) = L\elem_X(t,\sigma_{\qfun X}(f)) &\le (L\elem_X;LR)(t,\sigma_{\qfun X}(f)) \\
    &\le L\elem_X(t,\sigma_{\qfun X}(g)) \oplus LR(\sigma_{\qfun X}(g),\sigma_{\qfun X}(f))
      = \mu^\sigma_X(g)(t). \tag*{\qedhere}
  \end{align*}
\end{proof}

\paragraph{Preservation of nonexpansiveness}

\begin{lem}\label{lem:moss-pres-nonexp}
  Let $\sigma\in\Sigma$. Then the Moss lifting $\mu^\sigma$ preserves
  nonexpansiveness.
\end{lem}
\begin{proof}
  Let $R\colon\frel{A}{B}$ and consider the map $R[-]\colon\qfun
  A\to\qfun B, f\mapsto R[f]$. First, we show that
  \begin{equation}
    \elem_A \;\;\le\;\; R;\elem_B;\rev{\gr{R[-]}}. \label{eq:elem-r}
  \end{equation}
  Let $f\in\qfun A$, $g\in\qfun B$, and let $a\in A$, $b\in B$.
  We need to show that
  \begin{equation*}
    \elem_A(a,f) \le  R(a,b) \oplus \elem_B(b,g) \oplus \rev{\gr{R[-]}}(g,f).
  \end{equation*}
  If $g\neq R[f]$, this holds trivially as $\rev{\gr{R[-]}}(g,f) = 1$.
  Otherwise, if $g=R[f]$, then we have $f(a) \ominus R(a,b) \le g(b)$ by definition,
  and hence
  \begin{equation*}
    \elem_A(a,f) = f(a) 
      \le R(a,b) \oplus g(b) \le R(a,b) \oplus \elem_B(b,g) \oplus \rev{\gr{R[-]}}(g,f).
  \end{equation*}
  Now let $(f,g)$ be $R$-nonexpansive and let $t_1\in TA$ and
  $t_2\in TB$. We need to show that
  \begin{equation}
    \mu^\sigma_A(f)(t_1) - \mu^\sigma_B(g)(t_2) \le LR(t_1,t_2). \label{eq:moss-nonexp}
  \end{equation}
  By monotonicity of $\mu^\sigma$ and Lemma~\ref{lem:rel-nonexp-props}
  it is enough to show this for the case $g = R[f]$. In this case we
  have $TR[-](\sigma_{\qfun A}(f)) =\sigma_{\qfun B}(g)$ by naturality of $\sigma$.
  Applying the lax extension laws to \eqref{eq:elem-r}, we have
  $L\elem_A \le LR;L\elem_B;\rev{\gr{TR[-]}}$, so that
  \begin{align*}
    L\elem_A(t_1,\sigma_{\qfun A}(f))
      &\le LR(t_1,t_2) \oplus L\elem_B(t_2,\sigma_{\qfun B}(g)) \oplus \rev{\gr{TR[-]}}(\sigma_{\qfun B}(g),\sigma_{\qfun A}(f)) \\
      &=  LR(t_1,t_2) \oplus L\elem_B(t_2,\sigma_{\qfun B}(g)),
  \end{align*}
  and \eqref{eq:moss-nonexp} follows by rearranging.
\end{proof}

\paragraph{Separation} To show that $\Lambda$ is separating for $L$,
we need to make use of the fact that $L$ is finitarily separable.

\begin{lem}\label{lem:moss-separating}
  $\Lambda$ is separating for $L$, that is, $L \le K_\Lambda$.
\end{lem}
\begin{proof}
  Let $R\colon\frel{A}{B}$ and $t_1\in TA$, $t_2\in TB$. Let $\epsilon
  > 0$. Put $s\colon B\to\qfun A, s(b)(a) = R(a,b)$. Because the set
  of $\Sigma$-terms over $\qfun A$ generates $T_\omega\qfun A$ and $L$
  is finitarily separable, there exists some $\sigma\in\Sigma$ and
  some $f\in\qfun A$ such that we have $L\Delta_{\qfun A}(\sigma_{\qfun A}(f), Ts(t_2))
  \le \epsilon$ and $L\Delta_{\qfun A}(Ts(t_2), \sigma_{\qfun A}(f)) \le
  \epsilon$. Put $g = R[f]$. Then it suffices to show that
  \begin{equation} \label{eqn:moss-sep}
    \mu^\sigma_A(f)(t_1) - \mu^\sigma_B(g)(t_2) + 2\epsilon \ge LR(t_1,t_2).
  \end{equation}

  \noindent First, by construction and naturality
  (Lemma~\ref{lem:transform}),
  \begin{equation*}
    L\elem_A(t_1,Ts(t_2))
      = L(\elem_A\circ(\id_A\times s))(t_1,t_2) = LR(t_1,t_2),
  \end{equation*}
  where in the second step we used that $(\elem_A\circ(\id_A\times
  s))(a,b) = s(b)(a) = R(a,b)$ for all $a\in A,b\in B$. By the axioms of lax
  extensions we also have $L\elem_A \le L\elem_A;L\Delta_{\qfun A}$ and
  therefore
  \begin{equation*}
    LR(t_1,t_2) = L\elem_A(t_1,Ts(t_2))
      \le L\elem_A(t_1,\sigma_{\qfun A}(f)) \oplus
        L\Delta_{\qfun A}(\sigma_{\qfun A}(f),Ts(t_2))
      \le \mu^\sigma_A(f)(t_1) + \epsilon.
  \end{equation*}

  \noindent Second, by the axioms of lax extensions and using naturality again,
  \begin{equation}\label{eqn:separating-proof-1}
    \begin{split}
      L\elem_B(t_2,T(R[-]\circ s)(t_2))
      &= L(\elem_B\circ (\id_B\times(R[-]\circ s)))(t_2,t_2)\\
      &\le L\Delta_B(t_2,t_2)= \Delta_B(t_2,t_2)= 0,
    \end{split}
  \end{equation}
  where in the inequality we used that for all $b_1,b_2\in B$,
  \begin{align*} \textstyle \hspace{-3mm}
  (\elem_B\circ(\id_B\times(R[-]\circ s)))(b_1,b_2) = R[s(b_2)](b_1)
  = \sup_{a\in A} R(a,b_2) \ominus R(a,b_1) \le \Delta_B(b_1,b_2).
  \end{align*}
  As before, we have $L\elem_B \le L\elem_B;L\Delta_{\qfun B}$.
  We also have $\sigma_{\qfun B}(g) = TR[-](\sigma_{\qfun A}(f))$ by naturality of $\sigma$. 
  Therefore, by naturality of $L$:
  \begin{align*}
    &\mu^\sigma_B(g)(t_2) &&\\
    &=
    L\elem_B(t_2,\sigma_{\qfun B}(g)) &&\myby{definition}\\
    &\le L\elem_B(t_2,T(R[-]\circ s)(t_2)) \oplus
      L\Delta_{\qfun B}(T(R[-]\circ s)(t_2),TR[-](\sigma_{\qfun A}(f))) &&\myby{L2}\\
    &= L\Delta_{\qfun B}(T(R[-]\circ s)(t_2),TR[-](\sigma_{\qfun A}(f))) &&\myby{\ref{eqn:separating-proof-1}}\\
    &= L(\Delta_{\qfun B}\circ(R[-]\times R[-]))(Ts(t_2),\sigma_{\qfun A}(f)) &&\myby{naturality}\\
    &\le L\Delta_{\qfun A}(Ts(t_2),\sigma_{\qfun A}(f)) \le \epsilon. &&\myby{Lemma~\ref{lem:comp-graph}}
  \end{align*}
  Our target inequality \eqref{eqn:moss-sep} now follows by combining and
  rearranging the above inequalities.
\end{proof}

\paragraph{Nonexpansiveness}
We note that $\epsilon$-diagonals characterize the supremum
norm as follows:
\begin{lem}\label{lem:delta-eps-supnorm}
  Let $X$ be a set, let $f,g\colon X\to[0,1]$ and let $\epsilon > 0$.
  Then $\supnorm{f-g}\le\epsilon$ if and only if both $(f,g)$ and $(g,f)$ are
  $\Delta_{\epsilon,X}$-nonexpansive pairs.
\end{lem}

\begin{lem}\label{lem:moss-is-nonexp}
  Let $\sigma\in\Sigma$. If $L$ is nonexpansive, then the Moss lifting $\mu^\sigma$ is nonexpansive.
\end{lem}
\begin{proof}
  Let $f,g\in\qfun X$ with $\supnorm{f-g}\le\epsilon$.
  We need to show that $\supnorm{\mu^\sigma_X(f)-\mu^\sigma_X(g)}\le\epsilon$.
  By Lemma~\ref{lem:delta-eps-supnorm}, we know that the pairs
  $(f,g)$ and $(g,f)$ are $\Delta_{\epsilon,X}$-nonexpansive.
  Therefore, because the Moss liftings preserve nonexpansiveness
  (Lemma~\ref{lem:moss-pres-nonexp}), the pairs $(\mu^\sigma_X(f),\mu^\sigma_X(g))$
  and $(\mu^\sigma_X(g),\mu^\sigma_X(f))$ are $L\Delta_{\epsilon,X}$-nonexpansive,
  and thus they are also $\Delta_{\epsilon,TX}$-nonexpansive by (L4).
  The claim now follows by another application of Lemma~\ref{lem:delta-eps-supnorm}.
\end{proof}

\section{Real-valued Coalgebraic Modal Logic}\label{sec:cml}
\noindent We next recall the generic framework of \emph{real-valued
  coalgebraic modal logic}, which lifts two-valued coalgebraic modal
logic (Section~\ref{sec:prelim}) to the quantitative setting, and will
yield characteristic quantitative modal logics for all nonexpansive
lax extensions. The framework goes back to work on fuzzy description
logics~\cite{SchroderPattinson11}. The present version, characterized
by a specific choice of propositional operators, appears in work on
the coalgebraic quantitative Hennessy-Milner
theorem~\cite{km:metric-bisimulation-games}, and generalizes
quantitative probabilistic modal
logic~\cite{bw:behavioural-pseudometric}.

Given a set $\Lambda$ of nonexpansive (fuzzy) predicate liftings, the set
$\modf{\Lambda}$ of modal ($\Lambda$)-formulae is given by
\begin{equation}\label{eq:grammar}
  \phi,\psi ::= c \mid \phi\ominus c \mid \phi\oplus c \mid \phi\land\psi
  \mid \phi\lor\psi \mid \lambda(\phi_1,\dots,\phi_n)
\end{equation}
where $c\in\Rat$ and $\lambda\in\Lambda$ has arity $n$. The semantics
assigns to each formula~$\phi$ and each coalgebra $(A,\alpha)$ a
real-valued map
\begin{math}
  \Sem{\phi}_{A,\alpha}\colon A\to[0,1],
\end{math}
or just $\Sem{\phi}$, defined by
\begin{align*}
  \Sem{c}(a) &= c & \Sem{\phi\land\psi}(a) &= \min(\Sem{\phi}(a),\Sem{\psi}(a)) \\
  \Sem{\phi\ominus c}(a) &= \max(\Sem{\phi}(a)-c,0)
  & \Sem{\psi\lor\psi}(a) &= \max(\Sem{\phi}(a),\Sem{\psi}(a)) \\
  \Sem{\phi\oplus c}(a) &= \min(\Sem{\phi}(a)+c,1) & \Sem{\lambda(\phi_1,\dots,\phi_n)}(a) &=
  \lambda_A(\Sem{\phi_1},\dots,\Sem{\phi_n})(\alpha(a))
\end{align*}
\begin{rem}
  We thus adopt what is often called \emph{Zadeh semantics} for the
  propositional operators. This choice is pervasive in characteristic
  logics for behavioural distances
  (including~\cite{bw:behavioural-pseudometric,km:metric-bisimulation-games,WildEA18})
  -- in particular, the more general \emph{\L{}ukasiewiecz semantics}
  fails to be nonexpansive w.r.t.\ behavioural distance, and indeed
  induces a discrete logical distance~\cite{WildEA18}.
  
  In the same vein, we require the modalities $\lambda\in\Lambda$ to
  be nonexpansive to avoid situations where non-zero logical distances
  (Definition~\ref{def:logic-dist}) can be arbitrarily blown up by
  repeated application of modalities, such as in the case of the doubling
  modality $\lambda_X(f)(x) = 2f(x)$ of the identity functor. 

  In the two-valued setting, one can sometimes restrict the
  propositional base in characteristic logics; notably, two-valued
  probabilistic modal logic characterizes (event) bisimilarity of
  probabilistic transition systems even with conjunction as the only
  propositional connective~\cite{DesharnaisEA98}. No similar results
  appear to be known in the quantitative case; e.g.\ van Breugel and
  Worrell's characteristic logic for behavioural distance of
  probabilistic transition systems~\cite{bw:behavioural-pseudometric}
  does feature essentially the same propositional operators as our
  grammar~\eqref{eq:grammar}, if negation is defined as in
  Remark~\ref{rem:negation} below.

  Following~\cite{bw:behavioural-pseudometric}, we restrict truth
  constants in formulae to rational numbers, thus ensuring that the
  set of formulae is countable provided $\Lambda$ is countable. This
  countability is not needed for any of our results, and they will
  still hold if the truth constants come from any dense subset of
  $[0,1]$ (including $[0,1]$ itself).
\end{rem}

\begin{rem}\label{rem:negation}
  The logic as defined above does not include negation. This is to be
  expected, as already in the classical case the characteristic logic
  for similarity is negation-free modal logic with $\Diamond$ as the
  only modality~\cite{Glabbeek90}. However, if the set $\Lambda$ of
  predicate liftings is closed under duals (and the corresponding
  Kantorovich lifting therefore preserves converse), then negation
  $\neg\phi$ can be defined recursively using De Morgan's laws for the
  propositional operators and duals for the modalities:
  \begin{align*}
    \neg c &= 1-c &
    \neg(\phi\land\psi) &= \neg\phi\lor\neg\psi \\
    \neg (\phi\ominus c) &= \neg\phi\oplus c &
    \neg(\phi\lor\psi) &= \neg\phi\land\neg\psi \\
    \neg (\phi\oplus c) &= \neg\phi\ominus c &
    \neg\lambda(\phi_1,\dots,\phi_n) &= \bar\lambda(\neg\phi_1,\dots,\neg\phi_n)
  \end{align*}
  With negation defined like this, this version of real-valued coalgebraic modal
  logic is equivalent to the one in~\cite{ws:fuzzy-lax}, which includes negation
  as a primitive. The latter logic does not explicitly include addition, but in
  the presence of subtraction and negation we can define it as $\phi\oplus c =
  \neg(\neg\phi\ominus c)$.
\end{rem}

\begin{expl}\label{expl:logics}\hfill
  \begin{enumerate}
  \item \emph{Fuzzy modal logic}\label{item:fuzzy-as-coalg} may be
    seen as a basic fuzzy description
    logic~\cite{LukasiewiczStraccia08}. Eliding propositional atoms
    for brevity (they may be added as nullary modalities), we take
    $\Lambda=\{\Diamond,\Box\}$. Models are fuzzy relational structures,
    i.e.\ coalgebras for the \emph{covariant} fuzzy powerset
    functor~$\ffun$ given by $\ffun X = [0,1]^X$ and
    $\ffun f(g)(y) = \textstyle\sup_{f(x) = y}g(x)$, and~$\Diamond$ and $\Box$
    are interpreted as the predicate liftings
    \begin{equation*}
      \Diamond_A(f)(g) = \sup_{a\in A} \min(g(a),f(a))
      \qquad\text{and}\qquad
      \Box_A(f)(g) = \inf_{a\in A} \max(1-g(a),f(a)).
    \end{equation*}
    We note that $\Diamond$ and $\Box$ are dual, so that negation can
    be defined as in Remark~\ref{rem:negation}. Hennessy-Milner-type
    results necessarily apply only to finitely branching models, i.e.\
    coalgebras for~$\ffun_\omega$.
  \item \emph{Probabilistic modal logic:}\label{expl:prob-as-coalg}
    Take models to be probabilistic transition systems with possible
    deadlocks, i.e.\ coalgebras for the functor $1 + \dfun$, where
    $\dfun A$ is the set of discrete probability distributions on $A$
    (Section~\ref{sec:prelim}); and $\Lambda=\{\Diamond\}$, with
    \begin{equation*}
    \Diamond_A(f)(\ast) = 0 \quad\text{for $\ast\in1$, \qquad and}\qquad
    \Diamond_A(f)(\mu) = \expect{\mu}(f)
    = \textstyle\sum_{a\in A}\mu(a)\cdot f(a).
  \end{equation*}
  Probabilistic modal logic can be extended with negation by adding
  the dual $\Box$ of $\Diamond$. As taking expected values is
  self-dual, $\Box$ only differs from $\Diamond$ on deadlocks, where
  $\Box_A(f)(\ast) = 1$.
  When additionally extended with propositional atoms, this induces
  (up to restricting to discrete probabilities) van Breugel et al.'s
  contraction-free quantitative probabilistic modal
  logic~\cite{BreugelEA07}.
\end{enumerate}
\end{expl}

\noindent In the two-valued setting, modal logic is typically
invariant under bisimulation, i.e.\ bisimilar states satisfy the same
modal formulae. By contrast, under the asymmetric notion of
similarity, the corresponding statement is that the fragment of modal
logic that only uses the $\Diamond$ and no negations is
\emph{preserved under simulation}, i.e.\ if some state is simulated by
another state, then all formulae of this shape that are satisfied by
the first state are also satisfied by the second state.

In the quantitative setting, both of these statements correspond to
nonexpansiveness of formula evaluation w.r.t. the behavioural distance
arising from a Kantorovich lifting, where the distinction between the
two scenarios is embedded in the choice of modalities. We may also
phrase this more compactly by saying that logical distance is below
behavioural distance:
\begin{defn}\label{def:logic-dist}
  The \emph{$\Lambda$-logical distance} between states $a\in A$,
  $b\in B$ in $T$-coalgebras $(A,\alpha)$, $(B,\beta)$ is
  \begin{equation*}
    d^\Lambda(a,b) = \sup \{ \Sem{\phi}(a)\ominus\Sem{\phi}(b) \mid \phi\in\modf{\Lambda} \}.
  \end{equation*}
\end{defn}

\begin{lem}\label{lem:modal-logic-is-nonexp}
  Let~$\phi$ be a modal $\Lambda$-formula, and let $a\in A$, $b\in B$
  be states in $T$-coalgebras $(A,\alpha)$, $(B,\beta)$. Then
  \begin{equation*}
    \Sem{\phi}_{A,\alpha}(a)\ominus\Sem{\phi}_{B,\beta}(b) \le d^{K_\Lambda}_{\alpha,\beta}(a,b).
  \end{equation*}
\end{lem}

\begin{proof}
Induction on~$\phi$, with trivial Boolean cases (in Zadeh semantics,
all propositional operators on~$[0,1]$ are nonexpansive). For the
modal case, we have (for readability, restricting to unary~$\lambda\in\Lambda$)
\begin{align*}
  \Sem{\lambda(\phi)}(a)\ominus\Sem{\lambda(\phi)}(b)
  & = \lambda_A(\Sem{\phi})(\alpha(a))\ominus\lambda_B(\Sem{\phi}(\beta(b)) &&\myby{definition of $\Sem{\lambda(\phi)}$}\\
  & \le K_{\Lambda}d^{K_\Lambda}_{\alpha,\beta}(\alpha(a),\beta(b)) &&\myby{definition of $K_\Lambda$, IH}\\
  & = d^{K_\Lambda}_{\alpha,\beta}(a,b) && \myby{definition of $d^{K_\Lambda}_{\alpha,\beta}$}\hfill\qedhere
\end{align*}
\end{proof}

\begin{lem}[Nonexpansiveness of quantitative modal logic]\label{lem:nonexp-logic}\hfill \\
  If~$\Lambda$ preserves nonexpansiveness w.r.t.\ a lax
  extension~$L$, then $d^\Lambda\le d^L$.
\end{lem}
\begin{proof}
  Immediate from Lemma~\ref{lem:modal-logic-is-nonexp} and Lemma~\ref{lem:adequacy}.
\end{proof}

\noindent Finally, we show how the characterization of lax extensions as
Kantorovich extensions can be used to define characteristic logics for
nonexpansive lax extensions. Recall the sequence of approximants
(Definition~\ref{defn:approximant}) we used in Theorem~\ref{thm:fixpoint} to
approach the $L$-behavioural distance $d^L_{\alpha,\beta}$ of coalgebras
$\alpha\colon A\to TA$ and $\beta\colon B\to TB$ via fixpoint iteration:
\begin{equation*}\textstyle
  d_0 = 0, \quad\quad
  d_{n+1} = Ld_n \circ (\alpha\times\beta), \quad\quad
  d_\omega = \sup_{n<\omega} d_n.
\end{equation*}
If $L=K_\Lambda$, then each individual step in this iteration can be
related to the logical distance taken over some subset of
$\modf{\Lambda}$. More precisely, if we define the \emph{rank} of a
modal formula $\phi$ to be the maximal nesting depth of modalities,
then

\begin{lem}\label{lem:charact-finite-depth}
  For each $n<\omega$ and all $a\in A$, $b\in B$ we have:
  \begin{equation*}
    d_n(a,b) = \sup \{ \Sem{\phi}(a)\ominus\Sem{\phi}(b) \mid
      \phi\in\modf{\Lambda}, \text{ $\phi$ has rank at most $n$} \}.
  \end{equation*}
\end{lem}

\noindent A proof for the more general case of quantale-valued logics
and relations can be found in~\cite[Theorem 6.1]{ws:qc-benthem}. In
that paper, this characterization of finite-depth distances forms the
basis of a Hennessy-Milner theorem for the quantale-valued Kantorovich
lifting of finitary functors. In the present setting, we can drop the
condition that $T$ must be finitary by combining
Lemma~\ref{lem:charact-finite-depth} with Theorem~\ref{thm:fixpoint}
to obtain, complementing Lemma~\ref{lem:nonexp-logic}, a criterion
phrased directly in terms of conditions on the lax extension and the
modalities:
\begin{thm}[Coalgebraic quantitative Hennessy-Milner
  theorem]\label{thm:hm}
  Let~$L$ be a finitarily separable fuzzy lax extension, and let
  $\Lambda$ be a separating set of monotone nonexpansive predicate
  liftings for~$L$.  Then $d^\Lambda = d^L$.
\end{thm}
\begin{proof}
  By Lemma~\ref{lem:charact-finite-depth} we have $d_\omega =
  \sup_{n<\omega}d_n = d^\Lambda$ and by Theorem~\ref{thm:fixpoint} we
  have $d_\omega = d^L$.
\end{proof}
\begin{expl}\hfill
  \begin{enumerate}
    \item Since we only require~$L$ to be finitarily separable (rather
      than~$T$ finitary),
      Example~\ref{expl:kantorovich}.\ref{item:prob-kantorovich} implies
      that we recover
      expressiveness~\cite{BreugelEA07,bw:behavioural-pseudometric} of
      quantitative probabilistic modal logic over countably branching
      discrete probabilistic transition
      systems~(Example~\ref{expl:logics}.\ref{expl:prob-as-coalg}) as an
      instance of Theorem~\ref{thm:hm}.
    \item Let $L$ be the lax extension of $T = \Pow_\omega(M\times-)$
      from Example~\ref{expl:metric-ts}. As $T$ is finitary, it
      follows by Theorem~\ref{thm:kantorovich-lax} that $L=K_\Lambda$
      for the set $\Lambda$ of Moss liftings of $L$ and the logic
      $\modf{\Lambda}$ is characteristic for simulation distance
      by Theorem~\ref{thm:hm}.
  \end{enumerate}
\end{expl}

\noindent Applying Lemma~\ref{lem:nonexp-logic} and
Theorem~\ref{thm:hm} to~$L=K_\Lambda$ and using our result that all
lax extensions are Kantorovich extensions for their Moss liftings
(Theorem~\ref{thm:kantorovich-lax}), which moreover are monotone and
nonexpansive in case~$L$ is nonexpansive, we obtain expressive logics
for finitarily separable nonexpansive lax extensions:

\begin{cor}\label{cor:hennessy-milner}
  If $L$ is a finitarily separable nonexpansive lax extension of a
  functor~$T$, then $d^L = d^\Lambda$ for the set $\Lambda$ of Moss
  liftings.
\end{cor}
\noindent We can see the coalgebraic modal logic of Moss liftings as concrete
syntax for a more abstract logic where we incorporate functor elements
into the syntax directly, as in Moss' coalgebraic
logic~\cite{Moss99} and its generalization to lax extensions~\cite{MartiVenema15}. The set $\modf{L}$ of formulae in the arising
\emph{quantitative Moss logic} is generated by the same propositional
operators as above, and additionally by a modality $\Delta$ that
applies to $\Phi\in T\modf{0}$ for finite $\modf{0}\subseteq\modf{L}$,
with semantics
\begin{equation*}
  \Sem{\Delta\Phi}(a) = L\elem_A(\alpha(a),\Phi).
\end{equation*}
The dual of~$\Delta$ is denoted~$\nabla$, and behaves like a
quantitative analogue of Moss' two-valued~$\nabla$. From
Corollary~\ref{cor:hennessy-milner}, it is immediate that this logic
is expressive:
\begin{cor}[Expressiveness of quantitative Moss logic]\label{cor:moss-expr}
  Let~$L$ be a finitarily separable nonexpansive lax extension of a
  functor~$T$. Then $L$-behavioural distance $d^L$ coincides with
  logical distance in quantitative Moss logic, i.e.\ for all states
  $a\in A$, $b\in B$ in coalgebras $\alpha\colon A\to TA$, $\beta\colon B\to TB$,
  \begin{equation*}
    d^L_{\alpha,\beta}(a,b) = \sup \{ \Sem{\phi}(a)\ominus\Sem{\phi}(b) \mid \phi\in\modf{L} \}.
  \end{equation*}
\end{cor}

\begin{expl}\hfill
  \begin{enumerate}
  \item We equip the finite fuzzy powerset functor~$\ffun_\omega$ with
    the Wasserstein lifting $W_\Diamond$ for~$\Diamond$ as in
    Example~\ref{expl:logics}.\ref{item:fuzzy-as-coalg}, in analogy to
    the Hausdorff lifting
    (Example~\ref{expl:wasserstein}.\ref{item:hausdorff-is-lax}). Then
    $\nabla$ applies to finite fuzzy sets~$\Phi$ of formulae, and
    \begin{equation*}\textstyle
      \Sem{\nabla\Phi}(a)=
      \sup_{t\in\cpl{\Phi}{\alpha(a)}}\inf_{(\phi,a')\in \modf{L}\times A}\max(1-t(\phi,a'),\phi(a'))
    \end{equation*}
    for a state~$a$ in an $\ffun$-coalgebra $(A,\alpha)$, i.e.\ in a
    finitely branching fuzzy relational~structure.
  \item Let $C_{\mathsf{fg}}$ be the subfunctor of the convex
    powerset functor~$\Conv$ given by the finitely generated convex
    sets of (not necessarily finite) discrete distributions, equipped
    with the Wasserstein lifting described in
    Example~\ref{expl:wasserstein}.\ref{item:convex-powerset}.
    Then~$\nabla$ applies to finite sets of finite distributions on
    formulae, understood as spanning a convex polytope.  By
    Corollary~\ref{cor:moss-expr}, the arising instance of
    quantitative Moss logic is expressive for all
    $C_{\mathsf{fg}}$-coalgebras.
  \end{enumerate}
\end{expl}
\section{Conclusions}

We study behavioural distances based on fuzzy lax extensions, with a
particular focus on \emph{nonexpansive lax extensions}, establishing
that the latter are closely related to distances based on coalgebraic
modal logic. Nonexpansiveness of a lax extension can equivalently be
expressed in terms of strength of the underlying
functor~\cite{Gavazzo18} or as lax preservation of
$\epsilon$-diagonals.  We examine two general constructions of
nonexpansive lax extensions, respectively generalizing the classical
Kantorovich and Wasserstein distances and strengthening previous
generalizations where only pseudometrics are
lifted~\cite{bbkk:coalgebraic-behavioral-metrics}. Our construction of
the Kantorovich lifting is based in particular on generalizing
nonexpansive functions on a single space to nonexpansive pairs of
functions on two different spaces (implicit in work on optimal
transportation~\cite{Villani08}), while the Wasserstein lifting mostly
coincides with an existing construction from work on topological
theories~\cite{Hofmann07}.

Our main result shows that every
nonexpansive lax extension is a Kantorovich lifting for a suitable
choice of modalities, the so-called Moss modalities.  Moreover, one
can extract from a given nonexpansive lax extension a characteristic
modal logic satisfying a quantitative Hennessy-Milner property. Using
our notion of \emph{finitarily separable} lax extension additionally
allows us to extend these constructions to certain non-finitary
functors such as the discrete distribution functor. All our results
apply both to symmetric behavioural distances, i.e.~notions of
quantitative bisimulation, and to asymmetric behavioural distances,
i.e.~notions of quantitative simulation.

\bibliographystyle{alphaurl}
\bibliography{coalgml}

\end{document}